\documentclass[11pt]{article}

\usepackage{amsmath,amsfonts,amsthm}
\usepackage{soul}
\usepackage{xspace}
\usepackage{xcolor}
\usepackage{enumitem}
\usepackage{hyperref}
\usepackage[left=1in,right=1in,bottom=1in,top=1in]{geometry}
\usepackage{tabularx}
\usepackage{mdframed}

\setul{1ex}{.5pt}

\mdfdefinestyle{figstyle}{ %
  linecolor=black!7, %
  backgroundcolor=black!7, %
  innertopmargin=10pt, %
  innerleftmargin=25pt, %
  innerrightmargin=25pt, %
  innerbottommargin=10pt %
}

\newenvironment{protocol}{
  \begin{mdframed}[style=figstyle]}{
  \end{mdframed}}

\newtheorem{theorem}{Theorem}

\newtheorem{corollary}[theorem]{Corollary}

\theoremstyle{definition}
\newtheorem{action}{Group Action}
\newtheorem{definition}[theorem]{Definition}
\newtheorem{assumption}{Assumption}
\newtheorem{observation}{Observation}

\newtheorem{remark}[theorem]{Remark}

\newcommand{\abs}[1]{\left|{#1}\right|}

\newcommand{\mode}[1]{\textnormal{\textsf{#1}}}

\newcommand{\F}{\mathbb{F}}
\newcommand{\integer}{\mathbb{Z}}
\newcommand{\real}{\mathbb{R}}

\newcommand{\Matrix}{\mathrm{M}}
\newcommand{\Tensor}{\mathrm{T}}
\newcommand{\id}{\mathrm{id}}

\newcommand{\probsty}[1]{\textsf{#1}\xspace}
\newcommand{\GI}{\probsty{GI}}
\newcommand{\TI}{\probsty{TI}}
\newcommand{\DTI}{\probsty{DTI}}
\newcommand{\TTI}{\probsty{3TI}}
\newcommand{\gainv}{\probsty{GA-Inv}}
\newcommand{\gapr}{\probsty{GA-PR}}

\newcommand{\cA}{\mathcal{A}}
\newcommand{\cI}{\mathcal{I}}
\newcommand{\cX}{\mathcal{X}}
\newcommand{\cY}{\mathcal{Y}}

\newcommand{\cK}{\mathcal{K}}
\newcommand{\cP}{\mathcal{P}}
\newcommand{\cS}{\mathcal{S}}

\newcommand{\class}[1]{\ensuremath{\mathrm{#1}}\xspace}
\newcommand{\NP}{\class{NP}}

\newcommand{\coAM}{\class{coAM}}

\newcommand{\secpar}{\lambda}
\newcommand{\usecpar}{1^\secpar}

\newcommand{\bit}{\{0,1\}}

\newcommand{\abbrsty}[1]{\ensuremath{\mathrm{#1}}\xspace}
\newcommand{\OWA}{\abbrsty{OWA}}
\newcommand{\PRA}{\abbrsty{PRA}}
\newcommand{\PROD}{\abbrsty{PROD}}
\newcommand{\INV}{\abbrsty{INV}}
\newcommand{\QRO}{\abbrsty{QRO}}
\newcommand{\GLAT}{\abbrsty{GLAT}}

\newcommand{\pg}{\mathcal{G}}
\newcommand{\params}{\texttt{params}}
\newcommand{\attack}{\mathcal{A}}
\newcommand{\prover}{\mathcal{P}}
\newcommand{\verifier}{\mathcal{V}}
\newcommand{\hvsim}{\mathcal{S}}
\newcommand{\dualkey}{{\widetilde{pk}}}
\newcommand{\signlist}{\mathcal{L}}

\newcommand{\lencom}{{\ell_{\mathrm{in}}}}
\newcommand{\lench}{{\ell_{\mathrm{ch}}}}
\newcommand{\lenr}{{\ell_{\mathrm{re}}}}
\newcommand{\prep}{\ell}

\newcommand{\algstyle}[1]{\textsc{#1}\xspace}
\newcommand{\ids}{\algstyle{ID}}
\newcommand{\gaids}{\algstyle{GA-ID}}
\newcommand{\kg}{\algstyle{KG}}
\newcommand{\dkg}{\algstyle{KG}^*}
\newcommand{\skg}{\algstyle{KeyGen}}
\newcommand{\sign}{\algstyle{Sign}}
\newcommand{\vrfy}{\algstyle{Verify}}
\newcommand{\unruhsign}{\algstyle{SIGN}}
\newcommand{\gasign}{\algstyle{GA-SIGN}}
\newcommand{\fssign}{\algstyle{FS-SIGN}}
\newcommand{\gafssign}{\algstyle{GA-FS-SIGN}}
\newcommand{\prg}{\algstyle{PRG}}
\newcommand{\ggm}{\algstyle{GGM}}

\let\O\undefined
\let\S\undefined
\DeclareMathOperator{\O}{\mathrm{O}}
\DeclareMathOperator{\S}{\mathrm{S}}
\DeclareMathOperator{\adv}{\mathbf{Adv}}
\DeclareMathOperator{\GL}{\mathrm{GL}}
\DeclareMathOperator{\SL}{\mathrm{SL}}
\DeclareMathOperator{\poly}{\mathrm{poly}}
\DeclareMathOperator{\negl}{\mathrm{negl}}

\title{\bf General Linear Group Action on Tensors: A Candidate for Post-Quantum
  Cryptography}

\author{Zhengfeng Ji \thanks{ Centre for Quantum
    Software and Information, School of Software, Faculty of
    Engineering and Information Technology, University of Technology
    Sydney, NSW, Australia. {\tt Zhengfeng.Ji@uts.edu.au} } \and
  Youming Qiao \thanks{ Centre for Quantum Software and Information,
    School of Software, Faculty of Engineering and Information
    Technology, University of Technology Sydney, NSW, Australia. {\tt
      Youming.Qiao@uts.edu.au} } \and Fang Song \thanks{ Department of
    Computer Science and Engineering, Texas A\&M University, Texas,
    USA. {\tt fang.song@tamu.edu} } \and Aaram Yun \thanks{ Department
    of Cyber Security, Division of Software Science and Engineering,
    Ewha Womans University, Seoul, Korea. {\tt aaramyun@ewha.ac.kr} }
}

\date{}

\begin{document}
\clearpage\maketitle
\thispagestyle{empty}

\begin{abstract}

  Starting from the one-way group action framework of Brassard and
  Yung (Crypto '90), we revisit building cryptography based on group
  actions. Several previous candidates for one-way group actions no
  longer stand, due to progress both on classical algorithms (e.g.,
  graph isomorphism) and quantum algorithms (e.g., discrete
  logarithm).

  We propose the \emph{general linear group action on tensors} as a
  new candidate to build cryptography based on group actions. Recent
  works (Futorny--Grochow--Sergeichuk \emph{Lin.\ Alg.\ Appl.}, 2019) suggest that 
  the underlying algorithmic problem, the
  \emph{tensor isomorphism problem}, is the hardest one among several
  isomorphism testing problems arising from areas including coding theory, 
  computational group 
  theory, and multivariate cryptography.
  We present evidence to justify the viability of this proposal from
  comprehensive study of the state-of-art heuristic algorithms,
  theoretical algorithms and hardness results, as well as quantum
  algorithms.

  We then introduce a new notion called \emph{pseudorandom group
    actions} to further develop group-action based
  cryptography. Briefly speaking, given a group $G$ acting on a set
  $S$, we assume that it is hard to distinguish two distributions of
  $(s, t)$ either uniformly chosen from $S\times S$, or where $s$ is
  randomly chosen from $S$ and $t$ is the result of applying a random
  group action of $g\in G$ on $s$. This subsumes the classical
  decisional Diffie-Hellman assumption when specialized to a
  particular group action. We carefully analyze various attack
  strategies that support the general linear group action on tensors
  as a candidate for this assumption.

  Finally, we establish the quantum security of several cryptographic
  primitives based on the one-way group action assumption and the
  pseudorandom group action assumption.

\end{abstract}

\newpage
\setcounter{page}{1}
\section{Introduction}
\label{sec:intro}

Modern cryptography has thrived thanks to the paradigm shift to a
formal approach: precise \emph{definition} of security and
mathematically sound \emph{proof} of security of a given construction
based on accurate \emph{assumptions}.  Most notably, computational
assumptions originated from specific algebraic problem such as
factoring and discrete logarithm have enabled widely deployed
cryptosystems.

Clearly, it is imperative to base cryptography on diverse problems to reduce the
risk that some problems turn out to be easy.
One such effort was by Brassard and Yung soon after the early development of
modern cryptography~\cite{BY90}.
They proposed an approach to use a \emph{group action} to construct a
\emph{one-way function}, from which they constructed cryptographic primitives
such as bit commitment, identification and digital signature.
The abstraction of one-way group actions (\OWA) not only unifies the assumptions
from factoring and discrete logarithm, but more importantly Brassard and Yung
suggested new problems to instantiate it such as the graph isomorphism problem
(\GI).
Since then, many developments fall in this
framework~\cite{Pat96,Cou06,HS07,MRV07}.
In particular, the work of Couveignes~\cite{Cou06} can be understood as a
specific group action based on isogenies between elliptic curves, and it has
spurred the development of \emph{isogeny-based} cryptography~\cite{FJP14}.

However, searching for concrete group actions to support this approach turns out
to be a tricky task, especially given the potential threats from attackers
capable of quantum computation.
For graph isomorphism, there are effective heuristic solvers~\cite{McK80,MP14}
as well as efficient \emph{average-case} algorithms~\cite{BES80}, not to mention
Babai's recent breakthrough of a \emph{quasipolynomial}-time
algorithm~\cite{Bab16}.
Shor's celebrated work solves discrete logarithm and factoring in polynomial
time on a \emph{quantum} computer~\cite{Sho94}, which would break a vast
majority of public-key cryptography.
The core technique, \emph{quantum Fourier sampling}, has proven powerful and can
be applied to break popular symmetric-key cryptosystems as well~\cite{KLLN16}.
A \emph{subexponential-time} quantum algorithm was also found for computing
isogenies in ordinary curves~\cite{CJS14}, which attributes to the shift to
\emph{super-singular} curves in the recent development of isogeny-based
cryptography~\cite{GV18}.
In fact, there is a considerable effort developing \emph{post-quantum}
cryptography that can resist quantum attacks.
Besides {isogeny-based}, there are popular proposals based on discrete
\emph{lattices}, \emph{coding} problems, and \emph{multivariate
  equations}~\cite{BBD09,Chen16}.

\subsection{Overview of our results}

In this paper, we revisit building cryptography via the framework of group
actions and aim to provide new candidate and tools that could serve as
\emph{quantum-safe} solutions. Our contribution can be summarized in the
following three aspects.

First, we propose a family of group actions on \emph{tensors} of order at least
three over a finite field as a new candidate for one-way actions.
We back up its viability by its relation with other group actions, extensive
analysis from heuristic algorithms, provable algorithmic and hardness results,
as well as demonstrating its resistance to a standard quantum Fourier sampling
technique.

Second, we propose the notion of \emph{pseudorandom group actions}
(\PRA) that extends the scope of the existing group-action framework.
The \PRA assumption can be seen as a natural generalization of the
Decisional Diffie-Hellman (DDH) assumption.  We again instantiate it
with the group action on tensors, and we provide evidence (in addition
to those for one-wayness) from analyzing various state-of-art
attacking strategies.

Finally, based on any \PRA, we show realization of several primitives
in \emph{minicrypt} such as digital signatures via the Fiat-Shamir
transformation and pseudorandom functions. We give complete security
proofs against \emph{quantum} adversaries, thanks to recent advances
in analyzing quantum \emph{superposition} attacks and the quantum
random oracle model~\cite{Zhandry:2012ac,UnruhAC17,Song:2017aa}, which
is known to be a tricky business. 
Our constructions based on \PRA are
more efficient than known schemes based on one-way group actions. As a
side contribution, we also describe formal \emph{quantum-security}
proofs for several \OWA-based schemes including identification and signatures, 
which are missing in the literature
and deserve some care.

In what follows, we elaborate on our proposed group action based on
tensors and the new pseudorandom group action assumption. Readers interested in 
the cryptographic primitives supported by \PRA are referred to 
Section~\ref{sec:quantum_pra}.

\paragraph{The general linear group action on tensors.}
The candidate group action we propose is based on \emph{tensors}, a central
notion in quantum theory.
In this paper, a $k$-tensor $T$ is a multidimensional array with $k$ indices
$i_1, i_2, \ldots, i_k$ over a field $\F$, where $i_j \in \{1, 2, \ldots, d_j\}$
for $j=1, 2, \ldots, k$.
For a tuple of indices $(i_1, i_2, \ldots, i_k)$, the corresponding component of
$T$ denoted as $T_{i_1, i_2, \ldots, i_k}$ is an element of $\F$.
The number $k$ is called the order of the tensor.
A matrix over field $\F$ can be regarded as a tensor of order two.

We consider a natural group action on $k$-tensors that represents a local change
of basis.
Let $G = \prod_{j=1}^k \GL(d_j,\F)$ be the direct product of general linear
groups.
For $M = \bigl( M^{(j)} \bigr)_{j=1}^k \in G$, and a $k$-tensor $T$, the action
of $M$ on $T$ is given by
\begin{equation*}
  \alpha: (M, T) \mapsto \widehat{T}, \text{ where } \widehat{T}_{i_1, i_2, \ldots, i_k}
  = \sum_{l_1, l_2, \ldots, l_k} \biggl( \prod_{j=1}^k M^{(j)}_{i_j,l_j} \biggr)
  T_{l_1, l_2, \ldots, l_k}.
\end{equation*}
We shall refer to the above group action as the \emph{general linear
  group action on tensors} (\GLAT) of dimensions $(d_1, \dots, d_k)$
over $\F$, or simply $\GLAT$ when there is no risk of confusion. We
will consider group actions on tensors of order at least three, as the
problem is usually easy for matrices. In fact, in most of the cases,
we focus on $3$-tensors which is most studied and believed to be hard.

\paragraph{General linear actions on tensors as a candidate for one-way group
  actions.}

We propose to use \GLAT as an instantiation of one-way group actions.
Roughly speaking, a group action is called a \emph{one-way group action} (OWA in
short), if for a random $s\in S$, a random $g\in G$, $t=g\cdot s$, and any
polynomial-time adversary $\attack$ given $s$ and $t$ as input, $\attack$
outputs a $g'\in G$ such that $t=g'\cdot s$ only with negligible probability.

Breaking the one-wayness can be identified with solving some
isomorphism problem. Specifically, two $k$-tensors $T$ and
$\widehat{T}$ are said to be isomorphic if there exists an $M\in G$
such that $\widehat{T} = \alpha(M,T)$. We define the decisional tensor
isomorphism problem (\DTI) as deciding if two given $k$-tensors are
isomorphic; and the search version (\TI) is tasked with computing an
$M\in G$ such that $\widehat{T} = \alpha(M, T)$ if there is
one. Clearly, our assumption that \GLAT is a one-way group action is
equivalent to assuming that \TI is hard for random $M\in G$, random
$k$-tensor $S$, and $T := \alpha(M,S)$.  We focus on
the case when the order $k$ of the tensor equals three and the
corresponding tensor isomorphism problem is abbreviated as \TTI.  We
justify our proposal from multiple routes; see
Section~\ref{sec:owa_candidate} for a more formal treatment.

\begin{enumerate}
\item The $3$-tensor isomorphism problem can be regarded as ``the most
  difficult'' one among problems about testing isomorphism between objects, such
  as polynomials, graphs, linear codes, and groups, thanks to the recent work of
  Futorny, Grochow, and Sergeichuk~\cite{FGS19}.
  More specifically, it was proven in~\cite{FGS19} that several isomorphism
  problems, including graph isomorphism, quadratic polynomials with $2$ secrets
  from multivarite cryptography~\cite{Pat96}, $p$-group isomorphism from
  computational group theory~\cite{OBr94,LQ17}, and linear code permutation
  equivalence from coding theory~\cite{PR97,Sen00}, all reduce to \TTI; cf.
  Observation~\ref{obs:reductions}.
  Note that testing isomorphism of quadratic polynomials with two secrets has
  been studied in multivariate cryptography for more than two
  decades~\cite{Pat96}.
  Isomorphism testing of $p$-groups has been studied in computational group
  theory and theoretical computer science at least since the 1980's (cf.
  \cite{OBr94,LQ17}).
  Current status of these two problems then could serve as evidence for the
  difficulty of \TTI.

\item Known techniques that are effective on \GI, including the combinatorial
  techniques~\cite{WL68} and the group-theoretic techniques~\cite{Bab79,Luk82},
  are difficult to translate to \TTI.
  Indeed, it is not even clear how to adapt a basic combinatorial technique for
  \GI, namely individualizing a vertex \cite{BES80}, to the \TTI setting.
  It is also much harder to work with matrix groups over finite fields than to
  work with permutation groups.
  Also, techniques in computer algebra, including those that lead to the recent
  solution of isomorphism of quadratic polynomials with one secret~\cite{IQ18},
  seem not applicable to \TTI.

\item Finally, there is negative evidence that quantum algorithmic
  techniques involving the most successful quantum Fourier sampling
  may not be able to solve \GI and code
  equivalence~\cite{HMRRS10,DMR15}.
  It is expected that the same 
  argument
  holds with respect to \TTI as well. Loosely speaking, this is because the group 
  underlying
  \TTI is a direct product of general linear groups, which also has
  irreducible representations of high dimensions.

\end{enumerate}

\paragraph{A new assumption: pseudorandom group actions.}
Inspired by the Decisional Diffie-Hellman assumption, which enables
versatile cryptographic constructions, we propose the notion of
\emph{pseudorandom group actions}, or \PRA in short.

Roughly speaking, we call a group action $\alpha: G \times S \to S$
\emph{pseudorandom}, if any quantum polynomial-time algorithm
$\attack$ cannot distinguish the following two distributions except
with negligible probability: $(s, t)$ where $s, t\in_R S$, and the
other distribution $(s, \alpha(g, s))$, where $s\in_R S$ and
$g\in_R G$. A precise definition can be found in
Section~\ref{sec:pra}.

Note that if a group action is transitive, then the pseudorandom
distribution trivially coincides with the random distribution. Unless
otherwise stated, we will consider \emph{intransitive} group actions
when working with pseudorandom group actions. In fact, we can assume
that $(s, t)$ from the random distribution are in different orbits
with high probability, while $(s, t)$ from the pseudorandom
distribution are always in the same orbit.

Also note that \PRA is a stronger assumption than \OWA.  To break
\PRA, it is enough to solve the isomorphism testing problem \emph{on
  average} in a relaxed sense, i.e., on $1/\poly(n)$ fraction of the
input instances instead of all but $1/\poly(n)$ fraction, where $n$ is
the input size.

The Decisional Diffie-Hellman (DDH) assumption \cite{DH76,Bon98} can
be seen as the \PRA initiated with a certain group action; see
Observation~\ref{obs:classical_ddh}.  However, DDH is broken on a
quantum computer. We resort again to \GLAT as a quantum-safe candidate
of \PRA.  We investigate the hardness of breaking \PRA from various
perspectives and provide further justification for using the general
linear action on $3$-tensors as a candidate for \PRA.

\begin{enumerate}
\item Easy instances on $3$-tensors seem scarce, and average-case
  algorithms do not speed up dramatically. Indeed, the best known
  average-case algorithm, while improves over worst-case somewhat due
  to the birthday paradox, still inherently enumerate all vectors in
  $\F_q^n$ and hence take exponential time~\cite{BFV13,LQ17}.
\item For $3$-tensors, there have not been non-trivial and
  easy-to-compute isomorphism invariants, i.e., those properties that
  are preserved under the action. For example, a natural isomorphism
  invariant, the tensor rank, is well-known to be
  NP-hard~\cite{Has90}. Later work suggests that ``most
  tensor problems are NP-hard''~\cite{HL13}.
\item We propose and analyze several attack strategies from group theory and
  geometry. While effective on some
    non-trivial actions, these attacks do not work for the general
    linear action on $3$-tensors.
    For instance, we notice that breaking our \PRA from \GLAT reduces
    to the orbit closure intersection problem, which has received
    considerable attention in optimization, and geometric complexity
    theory. Despite recent
    advances~\cite{GCT5,BGO+18,BCG+18,AGL+18,DM18,IQS17}, any
    improvement towards a more effective attack would be a
    breakthrough.
  \end{enumerate}

Recently, De Feo and Galbraith proposed an assumption in the setting
of supersingular isogeny-based cryptography, which can be viewed as
another instantiation of \PRA~\cite[Problem 4]{FG18}. This gives more
reason to further explore \PRA as a basic building block in
cryptography.

\subsection{Discussions}

In this paper, we further develop and extend the scope of group action based
cryptography by introducing the general linear group actions on tensors, \GLAT,
to the family of instantiations, by formulating the pseudorandom assumption
generalizing the well-known DDH assumption, and by proving the quantum-security
of various cryptographic primitives such as signatures and pseudorandom
functions in this framework.

There are two key features of \GLAT that are worth mentioning
explicitly.  First, the general linear action is
\emph{non-commutative} simply because the general linear group is
non-abelian.  This is, on the one hand, an attractive property that
enabled us to argue the quantum hardness and the infeasibility of
quantum Fourier sampling type of attacks.  On the other hand, however,
this also makes it challenging to extend many attractive properties of
discrete-logarithm and decisional Diffie-Hellman to the more general
framework of group action cryptography.  For example, while it is
known that the worst-case DDH assumption reduces to the average-case
DDH assumption~\cite{NR04}, the proof relies critically on
commutativity.  Second, the general linear action is \emph{linear} and
the space of tensors form a linear space.  Linearity seems to be
responsible for the supergroup attacks on the $\PRA(d)$ assumption
discussed in Subsection~\ref{subsec:attacks}. It also introduces the
difficulty for building more efficient PRF constructions analogous to
the DDH-based ones proposed in~\cite{NR04}.

Our work leaves a host of basic problems about group action based
cryptography as future work.  First, we have been focusing on the
general linear group actions on tensors and have not discussed too
much about the other possible group actions on the tensor space.  A
mixture of different types of group actions on different indices of
the tensor may have the advantage of obtaining a more efficient
construction or other appealing structural properties.  It will be
interesting to understand better how the hardness of the group actions
on tensors relate to each other and what are the good choices of group
actions for practicability considerations.  Second, it is appealing to
recover the average-case to worst-case reduction, at least to some
extent, for the general group actions framework.  Finally, it is an
important open problem to build quantum-secure public-key encryption
schemes based on hard problems about \GLAT or its close variations.

\section{The group action framework}

\label{sec:framework}

In this section, we formally describe the framework for group action based
cryptography to be used in this paper. While such general frameworks were already 
proposed by Brassard and Yung~\cite{BY90} and Couveignes~\cite{Cou06}, there are
delicate differences in several places, so we will have to still go through the 
details. This section should be considered as largely expository. 

\subsection{Group actions and notations}

\label{subsec:group_prel}

Let us first formally define group actions.
Let $G$ be a group, $S$ be a set, and $\id$ the identity element of $G$.
A \emph{(left) group action} of $G$ on $S$ is a function $\alpha: G\times S\to
S$ satisfying the following: (1) $\forall s\in S$, $\alpha(\id, s)=s$; (2)
$\forall g, h\in G$, $s\in S$, $\alpha(gh, s)=\alpha(g, \alpha(h, s))$.
The group operation is denoted by $\circ$, e.g.
for $g, h\in G$, we can write their product as $g\circ h$.
We shall use $\cdot$ to denote the left action, e.g.
$g\cdot s=\alpha(g, s)$.
We may also consider the right group action $\beta:S\times G\to S$, and use the
exponent notation for right actions, e.g.
$s^g=\beta(s, g)$.

Later, we will use a special symbol $\bot\not\in G\cup S$ to indicate that a bit
string does not correspond to an encoding of an element in $G$ or $S$.
We extend the operators $\circ$ and $\cdot$ to $\circ:G\cup\{\bot\}\times
G\cup\{\bot\}\to G\cup\{\bot\}$ and $\cdot:G\cup\{\bot\}\times S\cup\{\bot\}\to
S\cup\{\bot\}$, by letting $g\circ h=\bot$ whenever $g=\bot$ or $h=\bot$, and
$g\cdot s=\bot$ whenever $g=\bot$ or $s=\bot$.

Let $\alpha:G\times S\to S$ be a group action.
For $s\in S$, the \emph{orbit} of $s$ is $O_s=\{t\in S : \exists g\in G, g\cdot
s=t\}$.
The action $\alpha$ partitions $S$ into a disjoint union of orbits.
If there is only one orbit, then $\alpha$ is called transitive.
Restricting $\alpha$ to any orbit $O$ gives a transitive action.
In this case, take any $s\in O$, and let $\mathrm{Stab}(s, G)=\{g\in G : g\cdot
s=s\}$ be the stabilizer group of $s$ in $G$.
For any $t\in O$, those group elements sending $s$ to $t$ form a coset of
$\mathrm{Stab}(s, G)$.
We then obtain the following easy observation.
\begin{observation}\label{obs:uniform}
Let $\alpha:G\times S\to S$, $s$, and $O$ be as above. The following two
distributions are the same: the uniform distribution of $t\in O$, and the
distribution of $g\cdot s$ where $g$ is sampled from a uniform distribution over
$G$.
\end{observation}

\subsection{The computational model}

\label{subsec:model}

For computational purposes, we need to model the algorithmic representations of
groups and sets, as well as basic operations like group multiplication, group
inverse, and group actions.
We review the group action framework as proposed in Brassard and
Yung~\cite{BY90}.
A variant of this framework, with a focus on restricting to abelian
(commutative) groups, was studied by Couveignes~\cite{Cou06}.
However, it seems to us that some subtleties are present, so we will propose
another version, and compare it with those by Brassard and Yung, and Couveignes,
later.

\begin{itemize}
\item Let $n$ be a parameter which controls the instance size. Therefore,
polynomial time or length in the following are with respect to $n$.

\item (Representing group and set elements.) Let $G$ be a group, and $S$ be a set.
Let $\alpha: G\times S\to S$ be a group action.
Group elements
and set elements are represented by bit strings $\{0, 1\}^*$. There are
polynomials $p(n)$ and $q(n)$,
such that we only work with group elements representable by $\{0, 1\}^{p(n)}$ and
set elements representable by $\{0, 1\}^{q(n)}$. There are
functions $F_G$ and $F_S$ from $\{0, 1\}^*$ to $G\cup \{\perp\}$ and $S\cup
\{\perp\}$, respectively. Here, $\perp$ is a special symbol, designating that the
bit string does not represent a group or set element. $F_G$ and $F_S$ should be
thought of as assigning bit strings to group elements.

\item (Unique encoding of group and set elements.)
For any $g\in G$, there exists a unique
$b\in \{0, 1\}^*$ such that $F_G(b)=g$. In particular, there exists a unique bit
string, also denoted by $\id$, such that $F_G(\id)=\id$. Similarly, for any $s\in
S$, there exists a unique $b\in \{0, 1\}^*$ such that $F_S(b)=s$.
\item (Group operations.) There are polynomial-time computable functions
$\PROD:\{0, 1\}^*\times\{0, 1\}^*\to\{0, 1\}^*$ and $\INV:\{0,
1\}^*\to \{0, 1\}^*$, such that for $b, c\in \{0, 1\}^*$,
$F_G(\PROD(b, c))=F_G(b)\circ F_G(c)$, and $F_G(\INV(b))\circ F_G(b)=\id$.
\item (Group action.) There is a polynomial-time function $a:\{0, 1\}^*\times \{0,
1\}^*\to \{0, 1\}^*$, such that for $b\in \{0, 1\}^*$ and $c\in \{0, 1\}^*$,
satisfies
$F_S(a(b, c))=\alpha(F_G(b), F_S(c))$.

\item (Recognizing group and set elements.)
  There are polynomial-time computable functions $C_G$ and $C_S$, such that
  $C_G(b)=1$ iff $F_G(b)\neq\bot$, and $C_S(b)=1$ iff $F_S(b)\neq\bot$.

\item (Random sampling of group and set elements.)
  There are polynomial-time computable functions $R_G$ and $R_S$, such that
  $R_G$ uniformly samples a group element $g\in G$, represented by the unique
  $b\in\{0, 1\}^{p(n)}$ with $F_G(b)=g$, and $R_S$ uniformly samples a set
  element $s\in S$, represented by some $b\in \{0, 1\}^{q(n)}$ with $F_S(b)=s$.
\end{itemize}

\begin{remark}\label{rem:model}
 
  Some remarks are due for the above model.

  \begin{enumerate}
  \item The differences with Brassard and Yung are: (1) allowing infinite groups
    and sets; (2) adding random sampling of set elements.
    Note that in the case of infinite groups and sets, the parameters $p(n)$ and
    $q(n)$ are used to control the bit lengths for the descriptions of
    legitimate group and set elements.
    This allows us to incorporate e.g.
    the lattice isomorphism problem \cite{HR14} into this framework.
    In the rest of this article, however, we will mostly work with finite groups
    and sets, unless otherwise stated.
  \item The main reason to consider infinite groups is the uses of lattice
    isomorphism and equivalence of integral bilinear forms in the cryptographic
    setting.

  \item The key difference with Couveignes lies in Couveignes's focus on
    transitive abelian group actions with trivial stabilizers.
  \item It is possible to adapt the above framework to use the black-box group
    model by Babai and Szemer\'edi~\cite{BS84}, whose motivation was to deal
    with non-unique encodings of group elements (like quotient groups).
    For our purposes, it is more convenient and practical to assume that the
    group elements have unique encodings.
  \item Babai~\cite{Babai:1991aa} gives an efficient Monte Carlo algorithm for
    sampling a group element of a finite group in a very general setting which
    is applicable to most of our instantiations with finite groups.
  \end{enumerate}
\end{remark}

\subsection{The isomorphism problem and the one-way assumption}
\label{subsec:owa}

Now that we have defined group actions and a computational model, let
us examine the isomorphism problems associated with group actions.

\begin{definition}[The isomorphism problem]\label{def:orbit}
Let $\alpha:G\times S\to S$ be a group action. The isomorphism problem for 
$\alpha$ is
to decide, given $s, t\in S$, whether $s$ and $t$ lie in the same orbit
under $\alpha$. If they are, the search version of the isomorphism problem further 
asks
to compute some $g\in G$, such that $\alpha(g, s)=t$.
\end{definition}

If we assume that there is a distribution on $S$ and we require the algorithm to
succeed for $(s, t)$ where $s$ is sampled from this distribution and $t$ is
arbitrary, then this is the \emph{average-case} setting of the isomorphism
problem.
For example, the first average-case efficient algorithm for the graph
isomorphism problem was designed by Babai, Erd\H{o}s and Selkow in the
1970's~\cite{BES80}.

The hardness of the isomorphism problem provides us with the basic intuition
for its use in cryptography.
But for cryptographic uses, the \emph{promised} search version of the
isomorphism problem is more relevant, as already observed by Brassard and
Yung~\cite{BY90}.
That is, suppose we are given $s, t\in S$ with the promise that they
are in the same orbit, the problem asks to compute $g\in G$ such that
$g\cdot s=t$.
Making this more precise and suitable for cryptographic purposes, we
formulate the following problem.

\begin{definition}[The group-action inversion (\gainv) problem] Let $\pg$ be a
  group action family, such that for a security parameter $\lambda$,
  $\pg(\usecpar)$ consists of descriptions of a group $G$, a set $S$ with
  $\log(|G|)=\poly(\lambda)$, $\log(|S|)=\poly(\lambda)$, and an group action
  $\alpha: G\times S \to S$ that can be computed efficiently, which we denote as
  a whole as a public parameter $\params$.
  Generate random $s\gets S$ and $g\gets G$, and compute $t: = \alpha(g,s)$.
  The group-action inversion (\gainv) problem is to find $g$ given $(s,t)$.
\label{def:ga-inv}
\end{definition}

\begin{definition}[Group-action inversion game]
  The group-action inversion game is the following game between a challenger and
  an arbitrary adversary $\attack$:
  \begin{enumerate}
  \item The challenger and adversary $\attack$ agree on the public parameter
    $\params$ by choosing it to be $\pg(\usecpar)$ for some security parameter
    $\secpar$.
  \item Challenger samples $s\gets S$ and $g\gets G$ using $R_S$ and $R_G$,
    computes $t = g \cdot s$, and gives $(s, t)$ to $\attack$.
  \item The adversary $\attack$ produces some $g'$ and sends it to the
    challenger.
  \item We define the output of the game $\gainv_{\attack,\pg}(\usecpar) = 1$ if
    $g' \cdot s = t$, and say $\attack$ wins the game if $\gainv_{\attack,
      \pg}(\usecpar) = 1$.
  \end{enumerate}
\end{definition}

\begin{definition} We say that the group-action inversion (\gainv) problem is
  hard relative to $\pg$, if for any polynomial time quantum algorithm $\attack$,
  \begin{equation*}
    \Pr \bigl[ \gainv_{\attack,\pg}(\usecpar) \bigr] \leq \negl(\secpar) \, .
  \end{equation*}
\end{definition}

We propose our first cryptographic assumption in the following. It
generalizes the one in~\cite{BY90}.

\begin{assumption}[One-way group action (\OWA) assumption]
  \label{asn:owa}
  There exists a family $\pg$ relative to which the $\gainv$ problem is hard.
\end{assumption}

We informally call the group action family $\pg$ in Assumption~\ref{asn:owa} a
one-way group action.
Its name comes from the fact that, as already suggested in~\cite{BY90}, this
assumption immediately implies that we can treat $\Gamma_s: G \to S$ given by
$\Gamma_s(g)=\alpha(g, s)$ as a one-way function for a random $s$.
In fact, OWA assumption is equivalent to the assertion that the function
$\Gamma:G\times S\to S\times S$ given by $\Gamma(g, s)=(g\cdot s, s)$ is one-way
in the standard sense.

Note that the \OWA assumption comes with the promise that $s$ and $t$ are in the
same orbit.
The question is to compute a group element that sends $s$ to $t$.
Comparing with Definition~\ref{def:orbit}, we see that the \OWA assumption is
stronger than the assumption that the search version of the isomorphism problem is
hard for a group action, while incomparable with the decision version.
Still, most algorithms for the isomorphism problem we are aware of do solve the 
search
version.

\begin{remark}\label{rem:compare_BY}
Note that Assumption~\ref{asn:owa} has a slight difference with that of Brassard
and Yung as follows.
In~\cite{BY90}, Brassard and Yung asks for the \emph{existence} of some $s\in S$
as in Definition~\ref{def:ga-inv}, such that for a random $g\in G$, it is not
feasible to compute $g'$ that sends $s$ to $\alpha(g, s)$.
Here, we relax this condition, namely a \emph{random} $s\in S$ satisfies this
already.
One motivation for Brassard and Yung to fix $s$ was to take into account of
graph isomorphism, for which Brassard and Crep\'eau defined the notion of ``hard
graphs'' which could serve as this starting point~\cite{BC86}.
However, by Babai's algorithm~\cite{Bab16} we know that hard graphs could not
exist.
Here we use a stronger notion by allowing a random $s$, which we believe is a
reasonable requirement for some concrete group actions discussed in
Section~\ref{sec:owa_candidate}.
\end{remark}

A useful fact for the \gainv problem is that it is \emph{self-reducible} to
random instances within the orbit of the input pair.
For any given $s$, let $O_s$ be the orbit of $s$ under the group action
$\alpha$.
If there is an efficient algorithm $\attack$ that computes $g$ from $(t, t')$ where
$t' = \alpha(g,t)$ for at least $1/\poly(\secpar)$ fraction of the pairs $(t,t')
\in O_s \times O_s$, then the \gainv problem can be computed for any $(t,t') \in
O_s \times O_s$ with probability $1-e^{-\poly(\secpar)}$.
On input $(t,t')$, the algorithm samples random group elements $h,h'$ and calls
$\attack$ with $(\alpha(h,t),\alpha(h',t'))$.
If $\attack$ successfully returns $g$, the algorithm outputs $h^{-1}gh'$ and
otherwise repeats the procedure for polynomial number of times.

The one-way assumption leads to several basic cryptographic
applications as described in the literature. First, it gives a
identification scheme by adapting the zero-knowledge proof system for
graph isomorphism~\cite{GMW91}. Then via the celebrated Fiat-Shamir
transformation~\cite{FS86}, one also obtains a signature
scheme. Proving quantum security of these protocols, however, would
need more care. For completeness, we give detailed proofs in
Section~\ref{sec:quantum_owa}.

\section{General linear actions on tensors: the one-way group action assumption}

\label{sec:owa_candidate}

In this section, we propose the general linear actions on tensors,
i.e., the tensor isomorphism problem, as our choice of candidate for
the \OWA assumption. We first reflect on what would be needed for a
group action to be a good candidate. 

\subsection{Requirements for a group action to be one-way}

\label{subsec:owa_requirement}

Naturally, the hardness of the $\gainv$ problem for a specific group action
needs to be examined in the context of the following four types of algorithms.
\begin{itemize}
\item Practical algorithms: implemented algorithms with practical performance
evaluations but no theoretical guarantees;
\item Average-case algorithms: for some natural distribution over the input
instances, there is an algorithm that are efficient for most input instances from
this distribution with provable guarantees;
\item Worst-case algorithms: efficient algorithms with provable guarantees for all
input instances;
\item Quantum algorithms: average-case or worst-case efficient algorithms in the
quantum setting.
\end{itemize}
Here, efficient means sub-exponential, and most means $1-1/\poly(n)$
fraction.  It is important to keep in mind all possible attacks by
these four types of algorithms.  Past experience suggests that one
problem may look difficult from one viewpoint, but turns out to be
easy from another.

The graph isomorphism problem has long been thought to be a difficult problem
from the worst-case viewpoint.
Indeed, a quasipolynomial-time algorithm was only known very recently, thanks to
Babai's breakthrough~\cite{Bab16}.
However, it has long been known to be effectively solvable from the practical
viewpoint~\cite{McK80,MP14}.
This shows the importance of practical algorithms when justifying a
cryptographic assumption.

Patarin proposed to use polynomial map isomorphism problems in his instantiation
of the identification and signature schemes~\cite{Pat96}.
He also proposed the one-sided version of such problems, which has been studied
intensively~\cite{PGC98,GMS03,Per05,FP06,Kay11,BFFP11,MPG13,BFV13,PFM14,BFP15},
mostly from the viewpoint of practical cryptanalysis.
However, the problem of testing isomorphism of quadratic polynomials with 
one secret was recently shown to be solvable in randomized polynomial
time~\cite{IQ18}, using ideas including efficient algorithms for computing the
algebra structure, and the $*$-algebra structure underlying such problems.
Hence, the investigation of theoretical algorithms is also valuable.

Considering of quantum attacks is necessary for
security in the quantum era. Shor's algorithm, for example, invalidates the 
hardness assumption of the
discrete logarithm problems.

Guided by the difficulty met by the hidden subgroup approach on tackling graph
isomorphism~\cite{HMRRS10}, Moore, Russell, and Vazirani proposed the code
equivalence problem as a candidate for the one-way assumption~\cite{MRV07}.
However, this problem turns out to admit an effective practical algorithm by
Sendrier~\cite{Sen00}.

\subsubsection{One-way group action assumption and the hidden subgroup approach}

\label{subsec:ow_hsp}

From the post-quantum perspective, a general remark can be made on the \OWA
assumption and the hidden subgroup approach in quantum algorithm design.

Recall that the hidden subgroup approach is a natural generalization of Shor's
quantum algorithms for discrete logarithm and factoring \cite{Sho94}, and can 
accommodate
both lattice problems \cite{Reg04} and isomorphism testing problems 
\cite{HMRRS10}. The survey paper of Childs and van Dam \cite{CD10} contains a nice 
introduction to this approach.

A well-known approach to formulate \gainv as an HSP problem is the following 
\cite[Sec. VII.A]{CD10}.
Let $\alpha:G\times S\to S$ be a group action.
Given $s, t\in S$ with the promise that $t=g\cdot s$ for some $g\in G$, we want
to compute $g$.
To cast this problem as an HSP instance, we first formulate it as an
automorphism type problem.
Let $\tilde G=G\wr \S_2$, where $\S_2$ is the symmetric group on two elements, and 
$\wr$ denotes the wreath product.
The action $\alpha$ induces an action $\beta$ of $\tilde G$ on $S\times S$ as
follows.
Given $(g, h, i)\in \tilde G=G\wr \S_2$ where $g, h\in G, i\in \S_2$, if $i$ is
the identity, it sends $(s, t)\in S\times S$ to $(g\cdot s, h\cdot t)$; otherwise,
it sends $(s, t)$ to $(h\cdot t, g\cdot s)$.
Given $(s, t)\in S\times S$, we define a function $f_{(s, t)}:\tilde G\to
S\times S$, such that $f_{(s, t)}$ sends $(g, h, i)$ to $(g, h, i)\cdot (s, t)$,
defined as above.
It can be verified that $f_{(s, t)}$ hides the coset of the stabilizer group of
$(s, t)$ in $\tilde G$.
Since $s$ and $t$ lie in the same orbit, any generating set of the stabilizer
group of $(s, t)$ contains an element of the form $(g, h, i)$, where $i$ is not
the identity element in $\S_2$, $g\cdot s=t$, and $h\cdot t=s$.
In particular, $g$ is the element required to solve the \gainv problem.
In the above reduction to the HSP problem, the ambient group is $G\wr \S_2$
instead of the original $G$.
In some cases like the graph isomorphism problem, because of the polynomial-time 
reduction from
isomorphism testing to automorphism problem, we can retain the ambient group to
be $G$.
However, such a reduction is not known for \GLAT.

There has been notable progress on the HSP problems for various ambient groups,
but the dihedral groups and the symmetric groups have withstood the attacks so
far.
Indeed, one source of confidence on using lattice problems in post-quantum
cryptography lies in the lack of progress in tackling the hidden subgroup
problem for dihedral groups~\cite{Reg04}.
There is formal negative evidence for the applicability of this approach
for certain group actions where the groups have high-dimensional
representations, like $\S_n$ and $\GL(n, q)$ in the case of the graph
isomorphism problem~\cite{HMRRS10} and the permutation code equivalence
problem~\cite{DMR15}.
The general lesson is that current quantum algorithmic technologies seem
incapable of handling groups which have irreducible representations of high
dimensions.

As mentioned, the \OWA assumption has been discussed in post-quantum
cryptography with the instantiation of the permutation code equivalence problem
\cite{MRV07,DMR11,DMR11b,SS13,DMR15}.
Though this problem is not satisfying enough due to the existence of effective
practical algorithms~\cite{Sen00}, the following quoted from~\cite{MRV07} would be
applicable to our choice of candidate to the discussed below.
\begin{quote}
  \it
  The design of efficient cryptographic primitives resistant to quantum attack
  is a pressing practical problem whose solution can have an enormous impact on
  the practice of cryptography long before a quantum computer is physically
  realized.
  A program to create such primitives must necessarily rely on insights into the
  limits of quantum algorithms, and this paper explores consequences of the
  strongest such insights we have about the limits of quantum algorithms.
\end{quote}

\subsection{The tensor isomorphism problem and others}

We now formally define the tensor isomorphism problem and other isomorphism
testing problems. For this we need some notation and preparations.

\subsubsection{Notation and preliminaries}

We usually use $\F$ to denote a field.
The finite field with $q$ elements and the real number field are denoted by
$\F_q$ and $\real$, respectively.
The linear space of $m$ by $n$ matrices over $\F$ is denoted by
$\Matrix(m,n,\F)$, and $\Matrix(n, \F):=\Matrix(n, n, \F)$.
The identity matrix in $\Matrix(n, \F)$ is denoted by $I_n$.
For $A\in \Matrix(m, n, \F)$, $A^t$ denotes the transpose of $A$.
The group of $n$ by $n$ invertible matrices over $\F$ is denoted by $\GL(n,
\F)$.
We will also meet the notation $\GL(n, \integer)$, the group of $n$ by $n$
integral matrices with determinant $\pm 1$.
We use a slightly non-standard notation $\GL(m, n, \F)$ to
denote the set of rank $\min(m, n)$ matrices in $\Matrix(m, n, \F)$.
We use $\langle \cdot \rangle$ to denote the linear span; for example, given
$A_1, \dots, A_k\in \Matrix(m, n, \F)$, $\langle A_1,\dots, A_k\rangle$ is
a subspace of $\Matrix(m, n, \F)$.

We will meet some subgroups of $\GL(n, \F)$ as follows.
The symmetric group $\S_n$ on $n$ objects is embedded into $\GL(n, \F)$ as
permutation matrices.
The orthogonal group $\O(n, \F)$ consists of those invertible matrices $A$ such
that $A^tA=I_n$.
The special linear group $\SL(n, \F)$ consists of those invertible matrices $A$
such that $\det(A)=1$.
Finally, when $n=\ell^2$, there are subgroups of $\GL(\ell^2, \F)$ isomorphic to
$\GL(\ell, \F)\times \GL(\ell, \F)$.
This can be seen as follows.
First we fix an isomorphism of linear spaces $\phi: \F^{\ell^2}\to \Matrix(\ell,
\F)$\footnote{For example, we can let the first $\ell$ components be the 
first row,
the second $\ell$ components be the second row, etc..}.
Then $\Matrix(\ell, \F)$ admits an action by $\GL(\ell, \F)\times \GL(\ell, \F)$
by left and right multiplications, e.g.
$(A, D)\in \GL(\ell, \F)\times \GL(\ell, \F)$ sends $C\in \Matrix(\ell, \F)$ to
$ACD^t$.
Now use $\phi^{-1}$ and we get one subgroup of $\GL(\ell^2, \F)$ isomorphic to
$\GL(\ell, \F)\times \GL(\ell, \F)$.

\subsubsection{Definitions of several group actions}
\label{subsec:group_actions}

We first recall the concept of tensors and the group actions on the space of
$k$-tensors as introduced in Section~\ref{sec:intro}.

\begin{definition}[Tensor]
  \label{def:tensor}
  A $k$-tensor $T$ of local dimensions $d_1, d_2, \ldots, d_k$ over $\F$,
  written as
  \begin{equation*}
    T = (T_{i_1, i_2, \ldots, i_k}),
  \end{equation*}
  is a multidimensional array with $k$ indices and its components $T_{i_1, i_2,
    \ldots, i_k}$ chosen from $\F$ for all $i_j \in \{1, 2, \ldots, d_j\}$.
  The set of $k$-tensors of local dimensions $d_1, d_2, \ldots, d_k$ over $\F$
  is denoted as
  \begin{equation*}
  \Tensor(d_1, d_2, \ldots, d_k, \F).
  \end{equation*}
  The integer $k$ is called the order of tensor $T$.
\end{definition}

\begin{action}[The general linear group action on tensors]
  \label{act:tensor}
  Let $\F$ be a field, $k$, $d_1, d_2, \ldots, d_k$ be integers.
  \begin{itemize}
  \item Group $G$: $\prod_{j=1}^k \GL(d_j, \F)$.
  \item Set $S$: $\Tensor(d_1, d_2, \ldots, d_k, \F)$.
  \item Action $\alpha$: for a $k$-tensor $T \in S$, a member $M = (M^{(1)},
    M^{(2)}, \ldots, M^{(k)})$ of the group $G$,
    \begin{equation*}
      \alpha(M, T) = \biggl( \bigotimes_{j=1}^k M^{(j)}  \biggr) T =
      \sum_{l_1, l_2, \ldots, l_k} \biggl( \prod_{j=1}^k M^{(j)}_{i_j,l_j} \biggr)
      T_{l_1, l_2, \ldots, l_k}.
    \end{equation*}
  \end{itemize}
\end{action}

We refer to the general linear group action on tensors in
Action~\ref{act:tensor} as \GLAT.
In the following, let us formally define several problems which have been
referred to frequently in the above discussions.

As already observed by Brassard and Yung \cite{BY90}, the discrete logarithm 
problem can be 
formulated using the language of group actions. More specifically, we have:
\begin{action}[Discrete Logarithm in Cyclic Groups of Prime
Orders]\label{act:discrete_log} Let $p$ be a prime, $\integer_p$ the integer.
\begin{itemize}
\item Group $G$: $\integer_p^{*}$, the multiplicative group of units in 
$\integer_p$.
\item Set $S$: $C_p\setminus \{\id\}$, where $C_p$ is a cyclic group
  of order $p$ and $\id$ is the identity element. 
\item Action $\alpha$: for $a\in \integer_p^{*}$, and $s\in S$,
$\alpha(a, s)=s^a$.
\end{itemize}
\end{action}
Note that in the above, we refrained from giving a specific realization of the
cyclic group $C_p$ for the sake of clarify; the reader may refer to Boneh's
excellent survey \cite{Bon98} for concrete proposals that can support the 
security of the Decisional Diffie-Hellman assumption.

The linear code permutation equivalence (LCPE) problem asks to decide whether
two linear codes (i.e.
linear subspaces) are the same up to a permutation of the coordinates.
It has been studied in the coding theory community since the
1990's~\cite{PR97,Sen00}.

\begin{action}[Group action for Linear Code Permutation Equivalence problem (LCPE)]
  \label{act:code}
  Let $m, d$ be integers, $m\leq d$, and let $\F$ be a field.
  \begin{itemize}
  \item Group $G$: $\GL(m, \F) \times \S_d$.
  \item Set $S$: $\GL(m, d, \F)$.
  \item Action $\alpha$: for $A \in S$, $M = (N,P) \in G$, $\alpha(M, A) = N A P^t $.
  \end{itemize}
\end{action}

The connection with coding theory is that $A$ can be viewed as the generating
matrix of a linear code (a subspace of $\F_q^n$), and $N$ is the change of basis
matrix taking care of different choices of bases.
Then, $P$, as a permutation matrix, does not change the weight of a codeword---
that is a vector in $\F^n$.
(There are other operations that preserve weights~\cite{SS13}, but we restrict
to consider this setting for simplicity.)
The \gainv problem for this group action is called the linear code permutation
equivalence (LCPE) problem, which has been studied in the coding theory
community since the 1980's~\cite{Leo82}, and we can dodge the only successful
attack~\cite{Sen00} by restricting to self-dual codes.

The following group action induces a problem called the polynomial isomorphism
problems proposed by Patarin~\cite{Pat96}, and has been studied in the
multivariate cryptography community since then.

\begin{action}[Group action for the Isomorphism of Quadratic Polynomials with
  two Secrets problem (IQP2S)]
  Let $m,d$ be integers and $\F$ a finite field.
  \label{act:poly}
  \begin{itemize}
  \item Group $G$: $\GL(d, \F) \times \GL(m, \F)$.
  \item Set $S$: The set of tuples of homogeneous polynomials $(f_1, f_2,
    \ldots, f_m)$ for $f_i \in \F[x_1, x_2, \ldots, x_d]$ the polynomial ring of
    $d$ variables over $\F$.
  \item Action $\alpha$: for $f = (f_1, f_2, \ldots, f_m) \in S$, $M = (C, D)
    \in G$, $C'=C^{-1}$, define $\alpha(M, f) = (g_1, g_2, \ldots, g_m)$ by 
    $g_i(x_1,  x_2, \ldots,
        x_d)=\sum_{j=1}^m D_{i,j} f_i(x_1', \dots, x_d')$, where 
        $x_i'=\sum_{j=1}^d C'_{i,j}x_j$.
  \end{itemize}
\end{action}

The \gainv problem for this group action is essentially the isomorphism of
quadratic polynomials with two secrets (IQP2S) assumption.
The algebraic interpretation here is that the tuple of polynomials $(f_1, \dots,
f_n)$ is viewed as a polynomial map from $\F^n$ to $\F^m$, by sending $(a_1,
\dots, a_n)$ to $(f_1(a_1, \dots, a_n), \dots, f_m(a_1, \dots, a_n))$.
The changes of bases by $C$ and $D$ then are naturally interpreted as saying
that the two polynomial maps are essentially the same.

Finally, the \gainv problem for the following group action originates from
computational group theory, and is basically equivalent to a bottleneck case of
the group isomorphism problem (i.e.~$p$-groups of class $2$ and exponent
$p$)~\cite{OBr94,LQ17}.

\begin{action}[Group action for alternating matrix space isometry (AMSI)]
\label{act:amsi}
  Let $d, m$ be integers and $\F$ be a finite field.
  \begin{itemize}
  \item Group $G$: $\GL(m, \F)$.
  \item Set $S$: the set of all linear spans $\cA$ of $d$ alternating\footnote{An 
  $m\times m$ matrix $A$ is alternating if for any $v\in \F^n$, $v^tAv=0$.} 
  matrices
    $A_i$ of size $m\times m$.
  \item Action $\alpha$: for $\cA = \langle A_1, A_2, \ldots, A_d \rangle \in
    S$, $C \in G$, $\alpha(C, \cA) = \langle B_1, B_2, \ldots, B_d \rangle$
    where $B_i = C A_i C^t$ for all $i=1, 2, \ldots, d$.
  \end{itemize}
\end{action}

\subsection{General linear actions on tensors as one-way action candidates}

\subsubsection{The central position of $3$-tensor isomorphism}

As mentioned, the four problems, linear code permutation equivalence
(LCPE), isomorphism of polynomials with two secrets (IQP2S), and
alternating matrix space isometry (AMSI), have been studied in coding
theory, multivariate cryptography, and computational group theory,
respectively, for decades.  Only recently we begin to see connections
among these problems which go through the \TTI problem thanks to the
work of Futorny, Grochow, and Sergeichuk \cite{FGS19}. We spell out
this explicitly.

\begin{observation}[\cite{FGS19,Gro12}]\label{obs:reductions}
IQP2S, AMSI, \GI, and LCPE reduce to \TTI.
\end{observation}
\begin{proof}
Note that the set underlying Group Action~\ref{act:amsi} consists of $d$-tuples of 
$m\times m$ alternating matrices. We can write such a tuple $(A_1, \dots, A_d)$ as 
a 3-tensor $A$ of dimension $m\times m\times d$, such that 
$A_{i,j,k}=(A_k)_{i,j}$. Then AMSI asks to test whether two such 3-tensors are in 
the same 
orbit under the action of $(M, N)\in \GL(m, \F)\times \GL(d, \F)$ by sending a 
3-tensor $A$ 
to the result of applying $(M, M, N)$ to $A$ as in the definition of \GLAT.

Such an action belongs to the class of actions on 3-tensors considered in 
\cite{FGS19} under 
the name \emph{linked actions}. 
This work constructs a function $r$ from 3-tensors to 3-tensors, 
such that $A$ and $B$ are in the same orbit under $\GL(m, \F)\times \GL(d, \F)$ if 
and only if $r(A)$ and $r(B)$ are in the same orbit under $\GL(m, \F)\times\GL(m, 
\F)\times\GL(d, \F)$. This function $r$ can be computed efficiently
\cite[Remark 1.1]{FGS19}. 

This explains the reduction of the isomorphism problem for Group 
Action~\ref{act:amsi} to the 3-tensor isomorphism problem. For Group 
Action~\ref{act:poly}, by using the classical correspondence between homogeneous 
quadratic polynomials and symmetric matrices, we can 
cast it in a form similar to 
Group Action~\ref{act:amsi}, and then apply the above reasoning using again 
\cite{FGS19}. 

Finally, to reduce the graph isomorphism problem (\GI) and the linear code 
permutation
equivalent problem (LCPE) to the 3-tensor isomorphism problem, we only need to 
take care of LCPE as \GI reduces to LCPE \cite{PR97}. To reduce LCPE to \TTI, we 
can reduce it to the matrix Lie algebra conjugacy problem by \cite{Gro12}, which 
reduces to \TTI by \cite{FGS19} along the linked action argument,  though this 
time linked in a different way.
\end{proof}

This put \TTI at a central position of these difficult
isomorphism testing problems arising from multivariate cryptography, computational 
group theory, and coding theory.
In particular, from the worst-case analysis viewpoint, \TTI is the hardest
problem among all these.
This also allows us to draw experiences from previous research in various
research communities to understand \TTI.

\subsubsection{Current status of the tensor isomorphism problem and its one-way
  action assumption}

\label{subsec:current_status}

We now explain the current status of the tensor isomorphism problem to support
it as a strong candidate for the \OWA assumption.
Because of the connections with isomorphism of polynomials with two secrets
(IQP2S) and alternating matrix space isometry (AMSI), we shall also draw results
and experiences from the multivariate cryptography and the computational group
theory communities.

For convenience, we shall restrict to finite fields $\F_q$, though other fields
are also interesting.
That is, we consider the action of $\GL(\ell, \F_q)\times \GL(n, \F_q) \times
\GL(m, \F_q)$ on $T \in \Tensor(\ell, n, m, \F_q)$.
Without loss of generality, we assume $\ell \geq n \geq m$.
The reader may well think of the case when $\ell = n = m$, which seems to be the
most difficult case in general.
Correspondingly, we will assume that the instances for IQP2S are $m$-tuples of
homogeneous quadratic polynomials in $n$ variables over $\F_q$, and the
instances for AMSI are $m$-tuples of alternating matrices of size $n\times n$
over $\F_q$.

To start, we note that \TTI over finite fields belongs to $\NP\cap\coAM$,
following the same $\coAM$-protocol for graph isomorphism. 

For the worst-case time complexity, it can be solved in time $q^{m^2}\cdot
\poly(\ell, m, n, \log q)$, by enumerating $\GL(m, q)$, and then solving an
instance of the matrix tuple equivalence problem, which asks to decide whether
two matrix tuples are the same under the left-right multiplications of invertible
matrices.
This problem can be solved in deterministic polynomial time by reducing
\cite{IQ18} to the module isomorphism problem, which in turn admits a
deterministic polynomial-time solution~\cite{CIK97,BL08,IKS10}.
It is possible to reduce the complexity to $q^{c m^2}\cdot \poly(\ell, m, n,
\log q)$ for some constant $0<c<1$, by using some dynamic programming technique
as in~\cite{LQ17}.
But in general, the worst-case complexity could not go beyond this at present,
which matches the experiences of IQP2S and AMSI as well; see~\cite{IQ18}.

For the average-case time complexity, it can be solved in time $q^{O(m)}\cdot
\poly(\ell, n)$, by adapting the average-case algorithm for AMSI in~\cite{LQ17}.
This also matches the algorithm for IQP2S which has an average-case running time
of $q^{O(n)}$~\cite{BFV13}.

For practical algorithms, we draw experiences from the computational group
theory community and the multivariate cryptography community.
In the computational group theory community, the current status of the art is
that one can hope to handle $10$-tuples of alternating matrices of size
$10\times 10$ over $\F_{13}$, but absolutely not, for $3$-tensors of local
dimension say $100$, even though in this case the input can still be stored in
only a few megabytes.\footnote{We thank James B.
  Wilson, who maintains a suite of algorithms for $p$-group isomorphism testing,
  for communicating this insight to us from his hands-on experience.
  We of course maintain responsibility for any possible misunderstanding, or
  lack of knowledge regarding the performance of other implemented algorithms.}
In the multivariate cryptography community, the Gr\"obner basis
technique~\cite{FP06} and certain combinatorial technique~\cite{BFV13} have been
studied to tackle IQF2S problem.
However, these techniques are not effective enough to break
it~\cite{BFV13}\footnote{In particular, as pointed out in \cite{BFV13}, one
  needs to be careful about certain claims and conjectures made in some
  literature on this research line.}.

For quantum algorithms, \TTI seems difficult for the hidden subgroup approach,
due to the reasons presented in Section~\ref{subsec:ow_hsp}.

Finally, let us also elaborate on the prospects of using those techniques for 
graph isomorphism~\cite{Bab16} and for isomorphism of quadratic polynomials with
one secret~\cite{IQ18} to tackle \TTI. In general, the difficulties of
applying these techniques seem inherent. 

We first check out the graph isomorphism side.
Recall that most algorithms for graph isomorphism, including Babai's
\cite{Bab16}, are built on two families of techniques: group-theoretic, and
combinatorial. To use the group-theoretic techniques, we need to work with 
matrix groups over finite fields instead of permutation groups. Algorithms for 
matrix groups over finite fields are in general far harder than those for 
permutation groups. For example, the basic membership problem is well-known to be 
solvable by Sims's algorithm \cite{Sim78}, while for matrix groups over finite 
fields of odd order, this was only recently shown to be efficiently solvable with 
a number-theoretic oracle and the algorithm is much more involved \cite{BBS09}. To 
use the combinatorial techniques, we need to work with linear or multilinear 
structures instead of combinatorial structures. This shift poses severe 
limitations on the use of most combinatorial techniques, like individualizing a
vertex. For example, it is quite expensive to enumerate all vectors in a vector 
space over a finite field, while this is legitimate to go over all elements in a 
set. 

We then check out the isomorphism of quadratic polynomials with one secret side. 
The techniques for settling this problem as in~\cite{IQ18} are based on those
developed for the module isomorphism problem~\cite{CIK97,BL08,IKS10}, involutive
algebras \cite{Wil09}, and computing algebra structures~\cite{FR85}.
The starting point of that algorithm solves an easier problem, namely testing
whether two matrix tuples are equivalent under the left-right multiplications.
That problem is essentially linear, so the techniques for the module isomorphism
problem can be used.
After that we need to utilize the involutive algebra structure~\cite{Wil09}
based on \cite{FR85}.
However, for \TTI, there is no such easier linear problem to start with, so it
is not clear how those techniques can be applied.

To summarize, the $3$-tensor isomorphism problem is difficult from all the four
types of algorithms mentioned in Section~\ref{subsec:owa_requirement}.
Furthermore, the techniques in the recent breakthrough on graph
isomorphism~\cite{Bab16}, and the solution of the isomorphism of quadratic
polynomials with one secret \cite{IQ18}, seem not applicable to this problem.
All these together support this problem as a strong candidate for the one-way
assumption.

\subsubsection{Choices of the parameters}

Having reviewed the current status of the tensor isomorphism problem, we lay out 
some principles of choosing the parameters for the security, namely the order 
$k$, the dimensions $d_i$, and the underlying field $\F$. 

Let us first explain why we focus on $k=3$, namely 3-tensors. Of course, $k$ needs 
to be $\geq 3$ as most problems about 2-tensors, i.e. matrices, are easy. We then 
note that there is certain evidence to support the possibility that the $k$-tensor 
isomorphism problem reduces to the 3-tensor isomorphism problem. That is, over 
certain fields, by \cite[Theorem 5]{AS05} and \cite{FGS19}, the degree-$k$ 
homogeneous form equivalence problem reduces to the 
the 3-tensor isomorphism problem in polynomial time. The former problem can be 
cast as an isomorphism problem for 
symmetric\footnote{A 
tensor $A=(A_{i_1, \dots, i_k})$ is symmetric if for any permutation 
$\sigma\in\S_k$, and any index $(i_1, \dots, i_k)$, $A_{i_1, \dots, 
i_k}=A_{i_{\sigma(1)}, \dots, i_{\sigma(k)}}$.} 
$k$-tensors under a certain action of $\GL$.
From the practical viewpoint though, it will be interesting to investigate into
the tradeoff between the local dimensions $d_i$ and $k$.

After fixing $k=3$, it is suggested to set $d_1=d_2=d_3$. This is because of the 
argument when examining the worst-case time complexity in the above subsection. 

Then for the underlying finite field $\F_q$, the intuition is that setting $q$ to 
be a large prime would be more secure. Note that we can still store an 
exponentially large prime using polynomially-many bits. 
This is because, if $q$ is small, then the ``generic'' behaviors as ensured by
the Lang--Weil type theorems \cite{LW54} may not be that generic. So some 
non-trivial 
properties may arise which then help with isomorphism testing. This is especially 
important for the pseudorandom assumption to be discussed Section~\ref{sec:pra}. 
We then examine 
whether we want to set $q$ to be a large prime, or a large field with a small 
characteristic. The former one is preferred, because the current techniques in 
computer algebra and computational group theory, cf. \cite{IQ18} and \cite{BBS09}, 
can usually 
work efficiently with large fields of small characteristics. 

However, let us emphasize that even setting $q$ to be a constant, we do not have
any concrete evidence for breaking \GLAT as a one-way group action candidate. That 
is, the 
above discussion on the field size issue is rather hypothetical and conservative.

\section{The pseudorandom action assumption}

\label{sec:pra}

In this section, we introduce the new security assumption for group actions,
namely pseudorandom group actions, which generalises the Decisional
Diffie-Hellman assumption.
In Section~\ref{sec:pra_candidate}, we shall study the prospect of using the
general linear action on tensors as a candidate for this assumption.
Then in Section~\ref{sec:quantum_pra}, we present the cryptographic uses of this
assumption including signatures and pseudorandom functions.

\begin{definition}
  \label{def:gapr}
  Let $\pg$ be a group family as specified before.
  Choose public parameters $\params = (G,S,\alpha)$ to be $\pg(\usecpar)$.
  Sample $s\gets S$ and $g\gets G$.
  The group action pseudorandomness (\gapr) problem is that given $(s, t)$,
  where $t = \alpha(g, s)$ or $t\gets S$, decide which case $t$ is sampled from.
\end{definition}

\begin{definition}[Pseudorandom group action game]
  The pseudorandom group action game is the following game between a challenger
  and an adversary $\attack$:
  \begin{itemize}
  \item The challenger and the adversary $\attack$ agree on the public
    parameters $\params = (G, S, \alpha)$ by choosing it to be $\pg(\usecpar)$
    for some security parameter $\secpar$.
  \item Challenger samples random bit $b\in \bit$, $s\gets S$, $g\gets G$, and
    chooses $t\gets S$ if $b=0$ and $t = g \cdot s$ if $b=1$.
  \item Give $(s, t)$ to $\attack$ who produces a bit $a\in \bit$.
  \item We define the output of the game $\gapr_{\attack,\pg}(\usecpar) = 1$ and
    say $\attack$ wins the game if $a=b$.
  \end{itemize}
\end{definition}

\begin{definition} We say that the group-action pseudorandomness (\gapr) problem
  is hard relative to $\pg$, if for any polynomial-time quantum algorithm
  $\attack$,
  \begin{equation*}
    \Pr [ \gapr_{\attack,\pg}(\usecpar) = 1 ] = \negl(\lambda).
  \end{equation*}
\end{definition}

Some remarks on this definition are due here.

\paragraph{For transitive and almost transitive actions.}
In the case of transitive group actions, as an easy corollary of
Observation~\ref{obs:uniform}, we have the following.
\begin{observation}
\gapr~problem is hard, if the group action $\alpha$ is transitive.
\end{observation}

Slightly generalizing the transitive case, it is not hard to see that
\gapr~problem is hard, if there exists a ``dominant'' orbit $O\subseteq S$.
Intuitively, this means that $O$ is too large such that random $s$ and $t$ from
$S$ would both lie in $O$ with high probability.
For example, consider the action of $\GL(n, \F)\times \GL(n, \F)$ on $\Matrix(n,
\F)$ by the left and right multiplications.
The orbits are determined by the ranks of matrices in $\Matrix(n, \F)$, and the
orbit of matrices of full-rank is dominant.
But again, such group actions seems not very useful for cryptographic purposes.
Indeed, we require the orbit structure to satisfy that random $s$ and $t$ do not
fall into the same orbit.
Let us formally put forward this condition.

\begin{definition}\label{def:dominant}
  We say that a group action $\alpha$ of $G$ on $S$ does not \emph{have a
    dominant orbit}, if
  \begin{equation*}
    \Pr_{s,t\gets S}\, [s, t \text{ lie in the same orbit}] = \negl(\lambda).
  \end{equation*}
\end{definition}

\begin{assumption}[Pseudorandom group action (\PRA) assumption]
  \label{asn:pra}
  There exists an $\pg$ outputting a group action without a dominant orbit,
  relative to which the $\gapr$ problem is hard.
\end{assumption}

The name comes from the fact that the PRA assumption says `in spirit' that the
function $\Gamma:G\times S\to S\times S$ given by $\Gamma(g, s)=(g\cdot s, s)$
is a secure PRG.
Here, it is only `in spirit', because the PRA assumption does not include the
usual expansion property of the PRG.
Rather, it only includes the inexistence of a dominant orbit.

\paragraph{Subsuming the classical Diffie-Hellman assumption.}

We now formulate the classical decisional Diffie-Hellman (DDH) assumption as an
instance of the pseudorandom group action assumption.
To see this, we need the following definition.

\begin{definition}
Let $\alpha:G\times S\to S$ be a group action. The \emph{$d$-diagonal action} of
$\alpha$, denoted by $\alpha^{(d)}$, is the group action of $G$ on $S^d$, the 
Cartesian product of $d$ copies of $S$, where
$g\in G$ sends $(s_1, \dots, s_d)\in S^d$ to $(g\cdot s_1, \dots, g\cdot s_d)$.
\end{definition}

The following observation shows that the classical DDH can be obtained by
instantiating \gapr~with a concrete group action.

\begin{observation}\label{obs:classical_ddh}
  Let $\alpha$ be the group action in Group Action~\ref{act:discrete_log}.
  The classical Decisional Diffie-Hellman assumption is equivalent to the \PRA
  assumption instantiated with $\alpha^{(2)}$, the $2$-diagonal action of
  $\alpha$.
\end{observation}
\begin{proof}
  Recall from Group Action~\ref{act:discrete_log} defines an action $\alpha$ of
  $G\cong \integer_p^*$ on $S=C_p\setminus \{\id\}$ where $C_p$ is a cyclic
  group of order $p$.
  The $2$-diagonal action $\alpha^{(2)}$ is defined by $a\in \integer_p^*$
  sending $(s, t)\in S\times S$ to $(s^a, t^a)$. Note that while $\alpha$ is 
  transitive, $\alpha^{(2)}$ is not, and in fact it does not have a dorminant 
  orbit.

  \PRA instantiated with $\alpha^{(2)}$ then asks to distinguish between the
  following two distributions.
  The first distribution is $((s, t), (s', t'))$ where $s, t, s', t'\in_R S$.
  Since $\alpha$ is transitive, by Observation~\ref{obs:uniform}, this
  distribution is equivalent to $((s, s^a), (s^b, s^c))$, where $s\in_R S$ and
  $a, b, c\in_R G$.
  The second distribution is $((s, t), (s^b, t^b))$, where $s, t\in_R S$, and
  $b\in_R G$.
  Again, by Observation~\ref{obs:uniform}, this distribution is equivalent to
  $((s, s^a), (s^b, s^{ab}))$, where $s\in_R S$, and $a, b\in_R G$.

  We then see that this is just the Decisional Diffie-Hellman
  assumption\footnote{Here we use the version of DDH where the generator of the
    cyclic group is randomly chosen as also used in \cite{CS98}.
    A recent discussion on distinction between fixed generators and random
    generators can be found in \cite{BMZ19}}.
\end{proof}

As will be explained in Section~\ref{subsec:pra_requirement}, the pseudorandom
assumption is a strong one, in a sense much stronger than the one-way
assumption.
Therefore, Observation~\ref{obs:classical_ddh} is important because, by casting
the classical Diffie-Hellman assumption as an instance of the pseudorandom
assumption, it provides a non-trivial and well-studied group action candidate
for this assumption.

Of course, the DDH assumption is no longer secure under quantum attacks.
Recently, this assumption in the context of supersingular isogeny based
cryptography has been proposed by De Feo and Galbraith in~\cite{FG18}.
We will study the possibility for the $3$-tensor isomorphism problem as a
pseudorandom group action candidate in Section~\ref{sec:pra_candidate}

\paragraph{The $d$-diagonal pseudorandomness assumption.}
Motivated by Observation~\ref{obs:classical_ddh}, it will be convenient to
specialize $\gapr$ to diagonal actions, and make the following assumption.

\begin{definition}
  The $d$-diagonal pseudorandomness ($\gapr(d)$) problem for a group action
  $\alpha$, is defined to be the pseudorandomness problem for the $d$-diagonal
  group action $\alpha^{(d)}$. 
\end{definition}

We emphasize that $\gapr(d)$ is just $\gapr$ applied to group actions of a
particular form, so a special case of $\gapr$.
Correspondingly, we define $\PRA(d)$ to be the assumption that $\gapr(d)$ is
hard relative to some $\pg$.

Given a group action $\alpha:G\times S\to S$, let $F_\alpha=\{f_g:S\to S \mid
g\in G, f_g(s)=g\cdot s\}$.
It is not hard to see that $\PRA(d)$ is equivalent to say that $F_\alpha$ is a
$d$-query weak PRF in the sense of Maurer and Tessaro~\cite{MT08}.
This gives a straightforward cryptographic use of the $\PRA(d)$ assumption.

Given $d, e\in \integer^+$, $d<e$, it is clear that $\PRA(e)$ is a stronger
assumption than $\PRA(d)$.
Indeed, given an algorithm $A$ that distinguishes between
$$((s_1, \dots, s_d), (g\cdot s_1,
\dots g\cdot s_d)) \text{ and } ((s_1, \dots, s_d), (t_1, \dots, t_d)),$$ where
$s_i, t_j\gets S$, and $g\gets G$, one can use $A$ to distinguish between
$((s_1, \dots, s_e), (g\cdot s_1, \dots g\cdot s_e))$ and $((s_1, \dots, s_e),
(t_1, \dots, t_e))$, by just looking at the first $d$ components in each tuple.

The applications of the $\PRA$ assumption including more efficient
quantum-secure digital signature schemes and pseudorandom function
constructions are given in Section~\ref{sec:quantum_pra}.
Next, we will provide candidates to instantiate the $\PRA$ assumption.

\section{General linear actions on tensors: the pseudorandom action assumption}

\label{sec:pra_candidate}

\subsection{Requirements for a group action to be pseudorandom}

\label{subsec:pra_requirement}

Clearly, a first requirement for a group action to be pseudorandom is that it
should be one-way.
Further requirements naturally come from certain attacks.
We have devised the following attack strategies.
These attacks suggest that the pseudorandom assumption is closely related to the
orbit closure intersection problem which has received considerable attention
recently.

\paragraph{Isomorphism testing in the average-case setting.}
To start with, we consider the impact of an average-case isomorphism testing
algorithm on the pseudorandom assumption.
Recall that for a group action $\alpha:G\times S\to S$, an average-case
algorithm is required to work for instances $(s, t)$ where $s\gets S$ and $t$ is
arbitrary.
Let $n$ be the input size to this algorithm.
The traditional requirement for an average-case algorithm is that it needs to
work for \emph{all but at most} $1/\poly(n)$ fraction of $s\in S$, like such
algorithms for graph isomorphism~\cite{BES80} and for alternating matrix space
isometry~\cite{LQ17}.
However, in order for such an algorithm to break the pseudorandom assumption, it
is enough that it works for a non-negligible, say $1/\poly(n)$, fraction of the
instances.
This is quite relaxed compared to the traditional requirement.

\paragraph{The supergroup attack.}

For a group action $\alpha:G\times S\to S$, a supergroup action of $\alpha$ 
is
another group action $\beta:H\times S\to S$, such that (1) $G$ is a subgroup of
$H$, (2) the restriction of $\beta$ to $G$, $\beta|_G$, is equal to $\alpha$.
If it further holds that (3.1) the isomorphism problem for $H$ is easy, and (3.2)
$\beta$ is not dominant, we will then have the following so-called
\emph{supergroup attack}.
Give input $s,t \in S$, the adversary for the $\gapr$ problem of $\alpha$ will
use the solver for the isomorphism problem for $H$ to check if $s, t$ are from the
same orbit induced by $H$ and return $1$ if they are from the same orbit and $0$
otherwise.
If $s, t$ are from the same orbit induced by $G$, the adversary always returns
the correct answer as $G$ is a subgroup of $H$.
In the case that $s, t$ are independently chosen from $S$, by the fact that
$\beta$ is not dominant, the adversary will return the correct answer $0$ with
high probability.

\paragraph{The isomorphism invariant attack.}
Generalizing the condition (3) above, we can have the following more general
strategy as follows.
We now think of $G$ and $H$ as defining equivalence relations by their orbit
structures.
Let $\sim_G$ (resp $\sim_H$) be the equivalence relation defined by $G$ (resp.
$H$).
By the conditions (1) and (2), we have (a) $\sim_H$ is coarser than $\sim_G$.
By the condition (3.1), we have (b) $\sim_H$ is easy to decide.
By the condition (3.2), we have (c) $\sim_H$ have enough equivalence classes.
Clearly, if a relation $\sim$, not necessarily defined by a supergroup $H$,
satisfies (a), (b), and (c), then $\sim$ can also be used to break the $\PRA$
assumption for $G$.

Such an equivalence relation is more commonly known as an isomorphism invariant,
namely those properties that are preserved under isomorphism.
The sources of isomorphism invariants can be very versatile.
The supergroup attack can be thought of as a special case of category where the
equivalence relation is defined by being isomorphic under a supergroup action.
Another somewhat surprising and rich ``invariant'' comes from geometry, as we
describe now.

\paragraph{The geometric attack.} 
In the case of matrix group actions, the underlying vector spaces usually come
with certain geometry which can be exploited for the attack purpose.
Let $\alpha$ be a group action of $G$ on $V\cong \F^d$.
For an orbit $O\subseteq V$, let its Zariski closure be $\overline O$.
Let $\sim$ be the equivalence relation on $V$, such that for $s, t\in O$, $s\sim
t$ if and only if $\overline{O_s} \cap \overline{O_t} \neq \emptyset$.
It is obvious that $\sim$ is a coarser relation than $\sim_G$.
Furthermore, except some degenerate settings when $m$ or $n$ are very small,
there would be enough equivalence classes defined by $\sim$, because of the
dimension reason.
So (a) and (c) are satisfied.
Therefore, if we could test efficiently whether the orbit closures of $s$ and
$t$ intersect, (b) would be satisfied and we could break the $\PRA$ for
$\alpha$.
This problem, known as the orbit closure intersection problem, has received
considerable attention recently.

Another straightforward approach based on this viewpoint is to recall that the
geometry of orbit closures is determined by the ring of invariant polynomials
\cite{MFK94}.
More specifically, the action of $G$ on $V$ induces an action on $\F[V]$, the
ring of polynomial functions on $V$.
As $V\cong \F^d$, $\F[V]\cong \F[x_1, \dots, x_d]$.
Those polynomials invariant under this induced action form a subring of $\F[V]$,
denoted as $\F[V]^G$.
If there exists one easy-to-compute, non-trivial, invariant polynomial $f$ from
$\F[V]^G$, we could then use $f$ to evaluate on the input instances and
distinguish between the random setting (where $f$ is likely to evaluate
differently) and the pseudorandom setting (where $f$ always evaluates the same).

\subsubsection{Example attacks}
\label{subsec:attacks}

We now list some examples to illustrate the above attacks. 

\paragraph{An example of using the isomorphism invariant attack.}
We first consider the isomorphism invariant attack in the graph isomorphism
case.
Clearly, the degree sequence, consisting of vertex degrees sorted from large to
small, is an easy to compute isomorphism invariant.
A brief thought suggests that this invariant is already enough to break the
pseudorandom assumption for graph isomorphism.

\paragraph{An example of using the geometric attack.}
We consider a group action similar to the $3$-tensor isomorphism case (Group
Action~\ref{act:tensor}), inspired by the quantum marginal
problem~\cite{BCG+18}.
Given a $3$-tensor of size $\ell\times n\times m$, we can ``slice'' this
$3$-tensor according to the third index to obtain a tuple of $m$ matrices of
size $\ell$ by $n$.
Consider the action of $G=\O(\ell, \F)\times \O(n, \F)\times \SL(m, \F)$ on
matrix tuples $\Matrix(\ell\times n, \F)^m$, where the three direct product
factors act by left multiplication, right multiplication, and linear combination
of the $m$ components, respectively.
For a matrix tuple $(A_1, \dots, A_m)$ where $A_i\in\Matrix(\ell\times n, \F)$,
form an $\ell n\times m$ matrix $A$ where the $i$-th column of $A$ is obtained by
straightening $A_i$ according to columns.
Then $A^tA$ is an $m$ by $m$ matrix.
The polynomial $f=\det(A^tA)$ is then a polynomial invariant for this action.
For this note that the group $\O(\ell, \F)\times \O(n, \F)$ can be embedded as a 
subgroup of
$\O(\ell n, \F)$, so its action becomes trivial on $A^tA$.
Then the determinant is invariant under the $\SL(m, \F)$.
When $m<\ell n$, which is the interesting case, $\det(A^tA)$ is non-zero.
It follows that we have a non-trivial, easy-to-compute, polynomial invariant
which can break the \PRA assumption for this group action.

\paragraph{An example of using the supergroup attack.}
We then explain how the supergroup attack invalidates the $\PRA(d)$ assumption
for certain families of group actions with $d>1$.

Let $\alpha$ be a linear action of a group $G$ on a vector space $V\cong
\F^{N}$.
We show that as long as $d>N$, $\PRA(d)$ does not hold.
To see this, the action of $G$ on $V$ gives a homomorphism $\phi$ from $G$ to
$\GL(V)\cong \GL(N, \F)$.
For any $g\in G$, and $v_1, \dots, v_d\in V$, we can arrange an $N\times d$
matrix $S=[v_1, \dots, v_d]$, such that $T=[\phi(g)v_1, \dots,
\phi(g)v_d]=\phi(g)[v_1, \dots, v_d]$.
On the other hand, for $u_1, \dots, u_d\in V$, let $T'=[u_1, \dots, u_d]$.
Let us consider the row spans of $S$, $T$ and $T'$, which are subspaces of
$\F^d$ of dimension $\leq N<d$.
Clearly, the row spans of $S$ and $T$ are the same.
On the other hand, when $u_i$'s are random vectors, the row span of $T'$ is
unlikely to be the same as that of $S$.
This gives an efficient approach to distinguish between $T$ and $T'$.

We can upgrade the above attack even further as follows.
Let
$\alpha$ be a linear action of $G$ on the linear space of matrices 
$M=\Matrix(m\times n,
\F)$. Recall that $\GL(m, \F)\times \GL(n, \F)$ acts on $M$ by left
and right multiplications. Suppose $\alpha$ gives rise to a homomorphism
$\phi:G\to \GL(m, \F)\times \GL(n, \F)$. For $g\in G$, if $\phi(g)=(A, B)\in
\GL(m, \F)\times \GL(n, \F)$, we let $\phi_1(g):=A\in \GL(m, \F)$, and
$\phi_2(g)=B\in \GL(n, \F)$. We now show that when $d>(m^2+n^2)/(mn)$, $\PRA(d)$
does
not hold for $\alpha$.
To see this, for any $g\in G$, and $S=(A_1, \dots, A_d)\in \Matrix(m\times n, \F)^d$,
let
$$T=(\phi_1(g)^tA_1\phi_2(g), \dots, \phi_1(g)^tA_d\phi_2(g)).$$
On the other hand,
let $T'=(B_1, \dots, B_d)\in M^d$. Since
$$\dim(S)=\dim(\GL(m\times n, \F)^d)=mnd>m^2+n^2=\dim(\GL(m, \F)\times \GL(n,
\F)),$$
$\alpha$ does not have a dominant orbit (cf. Definition~\ref{def:dominant}) This
means 
that, when $B_i$'s are sampled randomly
from $S$, $T'$ is unlikely to be in the same orbit as $S$. Now we use the fact
that, the isomorphism problem for the action of $\GL(m, \F)\times \GL(n, \F)$ on 
$S$ can
be solved in deterministic polynomial time
\cite[Proposition 3.2]{IQ18}. This gives an efficient approach to distinguish
between $T$ and $T'$.

Note that the set up here captures the Group Actions~\ref{act:code} and
\ref{act:poly} in
Section~\ref{subsec:group_actions}. 
For
example, suppose for Group Action~\ref{act:code}, we consider linear codes which
are $n/2$-dimensional subspaces of $\F_q^n$. Then we have $m=n/2$, so $\PRA(3)$
for this action does not hold, as $3>(m^2+n^2)/(mn)=5/2$.

On the other hand, when $d\leq (m^2+n^2)/(mn)$, such an attack may fail, simply
because of the existence of a dominant orbit.

\subsection{The general linear action on tensors as a pseudorandom action
  candidate}

We have explained why the general linear action on tensors is a good candidate
for the one-way assumption in Section~\ref{sec:owa_candidate}.
We now argue that, to the best of our knowledge, it is also a candidate for the
pseudorandom assumption.

We have described the current status of average-case algorithms for $3$-tensor 
isomorphism problem in Section~\ref{subsec:current_status}.
One may expect that, because of the relaxed requirement for the average-case
setting as discussed in Section~\ref{subsec:pra_requirement}, the algorithms 
in~\cite{LQ17,BFV13} may be
accelerated.
However, this is not the case, because these algorithms inherently
enumerate all vectors in $\F_q^n$, or improve somewhat by using the birthday
paradox. 

We can also explain why the relaxed requirement for the average-case setting is
still very difficult, by drawing experiences from computational group theory,
because of the relation between \GLAT and Group Action~\ref{act:amsi}, which in
turn is closely related to the group isomorphism problem as explained in
Section~\ref{subsec:group_actions}.
In group theory, it is known that the number of non-isomorphic $p$-groups of
class $2$ and exponent $p$ of order $p^\ell$ is bounded as
$p^{\frac{2}{27}\ell^3+\Theta(\ell^2)}$ \cite{BNV07}.
The relaxed average-case requirement in this case then asks for an algorithm
that could test isomorphism for a subclass of such groups containing
non-isomorphic groups as many as
$p^{\frac{2}{27}\ell^3+\Theta(\ell^2)}/\poly(\ell, \log
p)=p^{\frac{2}{27}\ell^3+\Theta(\ell^2)}$.
This is widely regarded as a formidable task in computational group theory: at
present, we only know of a subclass of such groups with $p^{O(\ell^2)}$ many
non-isomorphic groups that allows for an efficient isomorphism test \cite{LW12}.

The supergroup attack seems not useful here. The group $G=\GL(\ell, 
\F)\times\GL(n, \F)\times\GL(m, \F)$ naturally lives in $\GL(\ell n m, \F)$. 
However, by Aschbacher's classification of maximal subgroups of finite classical 
groups~\cite{Asc84}, there
are few natural supergroups of $G$ in $\GL(\ell n m, \F)$. The obvious ones 
include subgroups isomorphic to $\GL(\ell n, \F)\times \GL(m, \F)$, which is not 
useful 
because it has a dominant orbit (Definition~\ref{def:dominant}).

The geometric attack seems not useful here either.
The invariant ring here is trivial \cite{DW00}\footnote{If instead of $\GL(\ell,
  \F)\times\GL(n, \F)\times\GL(m, \F)$ we consider $\SL(\ell, \F)\times \SL(n,
  \F)\times\SL(m, \F)$, the invariant ring is non-trivial -- also known as the 
  ring of semi-invariants for the corresponding $\GL$ action -- but highly
  complicated.
  When $\ell=m=n$, we do not even know one single easy-to-compute non-trivial
  invariant.
  It further requires exponential degree to generate the whole invariant
  ring~\cite{DM19}.}.
For the orbit closure intersection problem, despite some recent exciting
progress  in~\cite{BGO+18,BCG+18,AGL+18,DM18,IQS17}, the current
best algorithms for the corresponding orbit closure intersection problems still
require exponential time.

Finally, for the most general isomorphism invariant attack, the celebrated paper
of Hillar and Lim~\cite{HL13} is just titled ``Most Tensor Problems Are
NP-Hard.''
This suggests that getting one easy-to-compute and useful isomorphism invariant for
\GLAT is already a challenging task. Here, useful means that the invariant does 
not lead to an equivalence relation with a dominant class in the sense of
Definition~\ref{def:dominant}.

The above discussions not only provide evidence for \GLAT to be pseudorandom, but 
also highlight how this problem connects to
various mathematical and computational disciplines.
We believe that this could serve a further motivation for all these works in
various fields.

\section{Primitives from the one-way assumption: proving quantum security}
\label{sec:quantum_owa}

Based on any one-way group action, it is immediate to derive a one-way
function family. A bit commitment also follows by standard techniques,
which we shall discuss later on in Section~\ref{sec:supp}.

In this section, we focus on the construction of a digital signature
that exists in the literature. It follows a very successful approach
of applying the Fiat-Shamir transformation on an identification
scheme. However, proving quantum security of this generic method turns
out to be extremely challenging and delicate. For this reason, we
include a complete description of the construction and a formal
security proof in the quantum setting.

\subsection{Identification: definitions}
\label{sec:owa_id-def}

An \emph{identification scheme} $\ids$ consists of a triple of
probabilistic polynomial-time algorithms\footnote{We only consider classical
  protocols where the algorithms can be efficiently realized on
  classical computers.} $(\kg, \prover,\verifier)$:

\begin{itemize}[label=$\bullet$]
\item Key generating: $(pk,sk)\gets \kg(\usecpar)$ generates a public
  key $pk$ and a secret key $sk$.
\item Interactive protocol: $\prover$ is given $(sk,pk)$, and
  $\verifier$ is only given $pk$. They then interact in a few rounds,
  and in the end, $\verifier$ outputs either 1 (i.e., ``accept'') or
  $0$ (i.e., ``reject'').
\end{itemize}

We assume that the keys are drawn from some relation, i.e.
$(sk,pk)\in R$, and $(\prover,\verifier)$ is an interactive protocol
to prove this fact. Let
$$L_R := \{pk: \exists\, sk \text{ such that } (sk,pk) \in R \} \, ,$$
be the set of valid public keys.

We will exclusively consider identification schemes of a special form,
terms as $\Sigma$ protocols. In a $\Sigma$ protocol
$(\prover,\verifier)$, three messages are exchanged in total:
\begin{itemize}[label=$\bullet$]
\item Prover's initial message: $I\gets \prover(sk,pk)$. $I$ is
  usually called the ``commitment'', and comes from $\bit^{\lencom}$.
\item Verifier's challenge: $c\gets \verifier(pk,I)$. Let the set of
  challenge messages (the challenge space) be $\bit^{\lench}$.
\item Prover's response: prover computes a response
  $r\in \bit^{\lenr}$ based on $(sk,I,c)$ and its internal state.
\end{itemize}
Here $\lencom,\lench,\lenr$ are interpreted as functions of the
security parameter $\secpar$. We omit writing $\secpar$ explicitly
for the ease of reading.

A basic requirement of an $\ids$ is the \emph{correctness}, which is
basically the completeness condition of the interactive protocol
$(\prover,\verifier)$ on correctly generated keys.

\begin{definition} An $\ids = (\kg,\prover,\verifier)$ is correct if
  \begin{equation*}
    \Pr \bigl[ \verifier(pk)=1 : (pk,sk)\gets \kg(\usecpar) \bigr]
    \geq 1 - \negl(\secpar).
  \end{equation*}

  \label{def:id-corr}
\end{definition}

Instead of defining security for $\ids$ directly, we usually talk
about various properties of the $\Sigma$ protocol associated with
$\ids$, which will make the $\ids$ useful in various applications. In
the following, we consider an adversary $\attack$ which operates and
shares states in multiple stages. Our notation does not explicitly
show the sharing of states though.

\begin{definition}(Adapting~\cite[Definition 4]{UnruhAC17}) Let
  $(\prover,\verifier)$ be a $\Sigma$ protocols with message length
  parameters $\lencom,\lench$ and $\lenr$.
  \begin{itemize}
  \item \textbf{Statistical soundness}: no adversary can generate a
    public key not in $L_R$ but manage to make $\verifier$
    accept. More precisely, for any algorithm $\attack$ (possibly
    unbounded),
    \begin{equation*}
      \begin{split}
        \Pr \bigl[ \verifier(pk,I,c,r) = 1 \wedge pk \notin L_R:&\\
        r \gets \attack(pk,I,c), & c\gets \verifier(pk, I), (pk,I)\gets
        \attack(\usecpar) \bigr] \leq \negl(\secpar).
      \end{split}
    \end{equation*}

  \item \textbf{Honest-verifier zero-knowledge (HVZK)}: there is a
    quantum polynomial-time algorithm $\hvsim$ (the simulator) such that for
    any quantum polynomial-time adversary $\attack$,
    \begin{equation*}
      \begin{split}
        \Bigl|\, & \Pr \bigl[ \attack(pk,I,c,r) = 1:
        (I,c,r) \gets (\prover(sk), \attack(pk)),
        (sk,pk)\gets \kg(\usecpar) \bigr] \\
        - & \Pr \bigl[ \attack(pk,I,c,r) = 1: (I,c,r) \gets \hvsim(pk),
        (sk,pk)\gets \kg(\usecpar) \bigr] \,\Bigr| \leq \negl(\secpar).
      \end{split}
    \end{equation*}

  \item \textbf{Computational special soundness}: from any two
    accepting transcripts with the same initial commitment, we can
    extract a valid secret key. Formally, there is a quantum polynomial-time
    algorithm $\mathcal{E}$ such that for any quantum polynomial-time $\attack$, we have
    that
    \begin{equation*}
      \begin{split}
        \Pr \bigl[ (sk, pk) \notin R\wedge \verifier(pk, I,c,r) \wedge
        \verifier(pk, I,c',r') =1 \wedge c\neq c': & \\
        (pk,I,c,r, c' , r' )\gets \attack(\usecpar), sk \gets \mathcal{E}(pk,
        I,c,r,c',r') \bigr] & \leq \negl(\secpar).
      \end{split}
    \end{equation*}
  \item \textbf{Unique response}: it is infeasible to find two
    accepting responses for the same commitment and challenge. Namely,
    for any quantum polynomial-time $\attack$,
    \begin{equation*}
      \Pr \bigl[ r\neq r'\wedge \verifier(pk,I,c,r) = 1 \wedge
      \verifier(pk,I,c,r') = 1 :
      (pk,I,c,r,r')\gets \attack(\usecpar) \bigr] \leq \negl(\secpar).
    \end{equation*}
  \item \textbf{Unpredictable commitment}: Prover's commitment has
    superlogarithmic collision-entropy.
    \begin{align*}
      \Pr \bigl[ c = c': c\gets \prover(sk,pk), c'\gets\prover(sk,pk),
      (sk,pk) \gets \kg(\usecpar) \bigr] \leq \negl(\secpar).
    \end{align*}
  \end{itemize}
 \label{def:sigma-prop}
\end{definition}

\subsection{Identification: construction from \OWA}
\label{sec:owa_id-cons}

We construct an $\ids$ based on the $\gainv$ problem. This is
reminiscent of the famous zero-knowledge proof system for
graph-isomorphism.

\begin{figure}[htb!]
  \centering
  \begin{protocol}
    \ul{Protocol $\ids$}\\[1em]
    Let $\pg$ be a family of group actions that is \OWA.
    Choose public parameters $\params: = (G,S,\alpha)$ to be $\pg(\usecpar)$.
    Construct $\ids = (\kg_0, \prover_0,\verifier_0)$ as follows:
    \begin{itemize}[label=$\bullet$]
    \item $\kg_0(\usecpar)$: Sample uniformly random $ s\gets S$ and $g\gets G$,
      and compute $t = \alpha(g,s)$.
      Output $pk := (s, t)$ and $sk := g$.
    \item $\prover_0$ commitment: $\prover_0$ picks a random $h\gets
      G$. Compute $ I = \alpha(h, s)$ and send to $\verifier_0$.
    \item $\verifier_0$ challenge: $c \gets \bit$.
    \item $\prover_0$ response: if $c = 0$, $\prover_0$ sends $r: = h$; if $c =
      1$, $\prover_0$ sends $r: = hg^{-1}$.
    \item $\verifier_0$ verdict: if $c = 0$, $\verifier_0$ outputs $1$
      iff. $\alpha(r, s) = I$; if $c = 1$, $\verifier_0$ outputs $1$
      iff. $\alpha(r,t) = I$.
    \end{itemize}
  \end{protocol}
  \caption{Identification protocol $\ids$ based on $\gainv$.}
  \label{fig:ids}
\end{figure}

To get a statistically sound protocol, we compose $\ids$'s interactive
protocol $(\prover_0,\verifier_0)$ in parallel
$\prep=\omega(\log \secpar)$ times. We denote the resulting protocol
$\gaids$.
\begin{figure}[htb!]
  \centering
  \begin{protocol}
    \ul{Protocol $\gaids$}\\[1em]
    Given $\ids = ( \kg_0, \prover_0,\verifier_0)$.
    Choose public parameters $\params: = (G,S,\alpha)$ to be $\pg(\usecpar)$.
    Construct $\gaids = (\kg, \prover,\verifier)$ as follows:
    \begin{itemize}[label=$\bullet$]
    \item $\kg(\usecpar)$: run $\kg_0$ independently $\ell$
      times. Output $pk := \{(s_i,t_i)\}_{i=1}^{\prep}$ and
      $sk := \{g_i\}_{i=1}^{\prep}$.
    \item $\prover$ commitment: for $i =1, \ldots, \prep$, run
      $\prover_0$ to produce $h_i\gets G$ and
      $ I_i = \alpha(h_i, s_i)$. Each time $\prover$ uses fresh
      randomness. Send $(I_i)_{i=1}^\prep$ to $\verifier$.
    \item $\verifier$ challenge: run $\verifier_0$ in parallel $\prep$
      times, i.e., $c \gets \bit^\prep$.
    \item $\prover$ response: for $ i =1, \ldots, \prep$, produce
      $r_i \gets \prover_0(c_i,sk,pk,h_i,I_i)$. Send $\verifier$
      $(r_i)_{i=1}^\prep$.
    \item $\verifier$ verdict: for $ i =1, \ldots, \prep$, run
      $b_i\gets \verifier_0(pk,I_i,c_i,r_i)$. Outputs $1$
      iff. all $b_i=1$.
  \end{itemize}
  \end{protocol}
  \label{fig:gaids}
  \caption{Identification protocol $\gaids$}
\end{figure}

\begin{theorem} $\gaids$ has \emph{correctness, statistically
    soundness, HVZK} and \emph{unpredictable commitment}, assuming
  Assumption~\ref{asn:owa} holds.
\label{thm:gaids}
\end{theorem}

\begin{proof} We prove the properties one by one.
  \begin{itemize}
  \item \textbf{Correctness}. This is clear.
  \item \textbf{Statistical soundness}. For any adversary $\attack$ who
    produces some $pk = (s,t)\notin L_R$, it implies that for any
    $i\in [\prep]$, it can only answer one of the two challenges
    ($c_i = 0$ or $1$) but not both. Since $c_i$'s are all uniformly
    chosen, $\verifier$ will reject except with probability
    $(\frac{1}{2})^{\prep} = \negl(\secpar)$ (noting that
    $\prep = \omega(\log\secpar)$).
  \item \textbf{HVZK}.  We construct a simulator $\hvsim$ in
    Fig.~\ref{fig:hvsim}.  The simulated transcript is identically
    distributed as the real execution.
    \begin{figure}[htb!]
      \centering
      \begin{protocol}
        Given $pk = \{(s_i,t_i)\}_{i=1}^{\prep}$, which is generated
        $(pk,sk) \gets \kg(\usecpar)$,
        \begin{itemize}
        \item $\hvsim$ generates $c \gets \bit^{\prep}$.
        \item For $i = 1, \ldots \prep$, if $c_i = 0$, let
          $I_i := \alpha(h_i, s_i)$ and $r_i : = h_i$; if $c_i=1$, let
          $I_i := \alpha(h_i, t_i)$ and $r_i : = h_i$.  Here
          $h_i \gets G $ are sampled randomly for all $i$.
        \item Output $(I,c,r) = (I_i,c_i,r_i)_{i=1}^\prep$.
        \end{itemize}
      \end{protocol}
      \caption{Simulator $\hvsim$}
      \label{fig:hvsim}
    \end{figure}
  \item \textbf{Special soundness}. If $\attack$ can produce $(I,c,r)$
    and $(I,c',r')$ that are both accepting with $c \neq c'$. Then at
    least $c_i \neq c_i'$ for some $i \in [\prep]$. The corresponding
    $r_i$ and $r_i'$ are hence $h$ and $hg^{-1}$ from which we can
    recover the secret key $g$.
  \item \textbf{Unpredictable commitment}. Two commitment messages
    collide only if $\alpha(g,s) = \alpha(g',s)$ for random
    $g,g'\gets G$. All our candidate group actions are (almost)
    injective.
  \end{itemize}
\end{proof}

\subsection{Digital signature from \OWA}
\label{sec:owa_sign}

\begin{definition} A \emph{digital signature scheme} consists of a
  triple of probabilistic polynomial-time algorithms $(\skg,\sign,\vrfy)$
  where
  \begin{itemize}[label=$\bullet$]
  \item $\skg$: $(pk,sk) \gets \skg(\usecpar)$ generates a pair of
    secret key and public key.
  \item $\sign$: on input $sk$ and message $m\in \mathcal{M}$, outputs
    $\sigma \gets \sign_{sk}(m)$.
  \item $\vrfy$: on input $pk$ and message-signature pair
    $(m,\sigma)$, output $\vrfy_{pk}(m,\sigma) = \text{acc/rej}$.
  \end{itemize}
\label{def:ds}
\end{definition}

A signature is secure if no one without the secret key can forge a
valid signature, even if it gets to see a number of valid
message-signature pairs. This is modeled as giving an adversary the
signing oracle. We give the formal definition below which explicitly
incorporates a random oracle $H$ that is available to all users, and
an adversary can access in quantum superposition. We stress that we do
not allow quantum access to the signing oracle, which is a stronger
attack model (cf.~\cite{BZC13}).

\begin{definition}[Unforgeability] A signature scheme
  $(\skg,\sign,\vrfy)$ is \emph{unforgeable} iff. for all quantum
  polynomial-time algorithm $\attack$,
  \begin{equation*}
    \begin{split}
      \Pr \bigl[ \vrfy^H(pk,\sigma^*,m^*) = 1 \wedge m^*\notin
      \signlist: & \\
      (pk,sk)\gets \skg(\usecpar), & (m^*,\sigma^*) \gets
      \attack^{H,\sign_{sk}}(\secpar,pk) \bigr]\leq \negl(\secpar).
    \end{split}
  \end{equation*}
  Here $\signlist$ contains the list of messages that $\attack$ queries
  to the (classical) signing oracle $\sign_{sk}(\cdot)$.
  \label{def:uf}
\end{definition}

Note the unforgeability does not rule out an adversary that produces a
new signature on some message that it has queried before. Strong
unforgeability takes this into account.

\begin{definition}[Strong Unforgeability] A signature scheme
  $(\skg,\sign,\vrfy)$ is \emph{strongly unforgeable} iff. for all quantum
  polynomial-time algorithm $\attack$,
  \begin{equation*}
    \begin{split}
      \Pr \bigl[ \vrfy^H(pk,\sigma^*,m^*) = 1 \wedge (m^*,\sigma^*) \in
      \signlist: & \\
      (pk,sk) \gets \skg(\usecpar), (m^*,\sigma^*) & \gets
      \attack^{H,\sign_{sk}}(\secpar,pk) \bigr] \leq \negl(\secpar).
    \end{split}
  \end{equation*}
  Here $\signlist$ contains the list of message \emph{and signature}
  pairs that $\attack$ queries to the (classical) signing oracle
  $\sign_{sk}(\cdot)$.
  \label{def:suf}
\end{definition}

Fiat and Shamir proposed a simple, efficient, and generic method that
converts an identification scheme to a signature scheme using a hash
function, and the security can be proven in the random oracle
model~\cite{FS86,PS00}.
Relatively recent, Fischlin proposed a variant to partly reduce the
reliance on the random oracle~\cite{Fis05}.
However, as shown in~\cite{ARU14}, both of them seem difficult to
admit a security proof in the quantum setting.
Instead, Unruh~\cite{UnruhEC15} proposed a less efficient
transformation, hereafter referred to as Unruh transformation, and
proved its security in the quantum random oracle model.
Our $\gaids$ satisfies the conditions required in Unruh transformation,
and hence we can apply it and obtain a digital signature scheme $\gasign$.

\begin{figure}[htb!]
  \centering
  \begin{protocol}
    \ul{Signature scheme $\unruhsign$ based on Unruh transformation}\\[1em]
    Let $\secpar$ be the security parameter.
    Let $t$ and $s$ be integers such that $t\log s = \omega(\log \secpar)$.
    Let $\lencom,\lench,\lenr$ be the length of the commitment, challenge and
    response respectively.
    Choose hash functions $H_1: \bit^* \to \{1,\ldots, s\}^t$ and $H_2:
    \bit^{\lenr}\to \bit^{\lenr}$.
    For an identification scheme $\ids$, we construct a signature scheme
    $\unruhsign$ as follows:
    \begin{itemize}
    \item $\skg$: run $\kg(\usecpar)$ of $\ids$ to obtain
      $(pk,sk)$.
    \item $\sign$: for input $sk$ and message $m$ do the following
      \begin{enumerate}[label=\roman*)]
      \item for $i = 1,\ldots, t$ and $j = 1,\ldots, s$, generate
        $I_i \gets \prover(pk,sk)$,
        $c_{i,j}\gets \bit^{\lench}\backslash
        \{c_{i,1},\ldots,c_{i,j-1}\}$. Then compute
        $r_{i,j} \gets \prover(c_{i,j},sk)$ and
        $h_{i,j} := H_2(r_{i,j})$.
      \item compute
        $ H_1(pk,m, (I_i)_{i=1,\ldots,t},
        (c_{i,j},h_{i,j})_{i=1,\ldots,t,j=1,\ldots,s})$. Partition
        the output as $J_1\| \ldots \| J_t$ where each
        $J_i \in [s]$.
      \item output
        $\sigma:= ((I_i)_{i =1,\ldots,t},
        (c_{i,j},h_{i,j})_{i=1\,\ldots,t,j=1,\ldots,s},
        (r_{i,J_i})_{i=1,\ldots,t})$.
      \end{enumerate}
    \item $\vrfy$: for input $(m,\sigma)$ and public key $pk$,
      \begin{enumerate}[label=\roman*)]
      \item parse $\sigma$ as
        $((I_i)_{i=1,\ldots,t},
        (c_{i,j},h_{i,j})_{i=1,\ldots,t,j=1,\ldots,s},
        (r_i)_{i=1,\ldots,t})$.
      \item compute
        $J_1\| \ldots \| J_t := H_1(pk,m, (I_i)_{i=1,\ldots,t},
        (c_{i,j},h_{i,j})_{i=1,\ldots,t,j=1,\ldots,s})$.
      \item for $i=1,\ldots ,t$, check $c_{i,1},\ldots, c_{i,s}$
        distinct; check $\verifier(pk, I_i,c_{i,J_i}, r_i) = 1$;
        check $h_{i,J_i} = H_2(r_i)$.
      \item accept if all checks pass.
      \end{enumerate}
    \end{itemize}
  \end{protocol}
  \caption{Unruh transformation}
  \label{fig:id-sign}
\end{figure}

\begin{theorem}[Adapting Corollary 19 \& Theorem 23 \& of~\cite{UnruhEC15}] If
  an identification scheme $\ids$ have correctness, HVZK and special soundness,
  then protocol $\unruhsign$ in Fig.~\ref{fig:id-sign} is a strongly unforgeable
  signature in \QRO.
\label{thm:id-sign}
\end{theorem}

Since $\gaids$ has these properties, we instantiate $\unruhsign$ with $\gaids$
and call the resulting signature scheme $\gasign$.
\begin{corollary} If Assumption~\ref{asn:owa} holds, $\gasign$ is a
  strongly unforgeable signature in \QRO.
  \label{cor:owa_sign}
\end{corollary}

\section{Quantum-secure primitives from the pseudorandom action
  assumption}
\label{sec:quantum_pra}

\subsection{Improved Digital signature based on \PRA via Fiat-Shamir}
\label{subsec:pra_sign}

In this subsection, we show that if we accept the possibly stronger
assumption of \PRA, we can apply the standard Fiat-Shamir transform
to $\gaids$, and obtain a more efficient signature scheme in
\QRO. This is due to a recent work by
Unruh~\cite{UnruhAC17}\footnote{\cite{KLS18} includes a similar
  result, which is primarily tailored to lattice-based identification
  schemes.}, where he shows that if one can equip $\ids$ with a
``dual-mode'' key generation, Fiat-Shamir will indeed work in \QRO. A
dual-mode key is a fake public key $\dualkey$ that is
indistinguishable from a valid public key. Nonetheless, $\dualkey$ has
no corresponding secret key (i.e., $\dualkey\notin L_R$).

\begin{definition}[Dual-mode key generator, adapting~\cite{UnruhAC17}]
  An algorithm $\kg$ is a \emph{dual-mode key generator} for a relation
  $R$ iff.
  \begin{itemize}[label=$\bullet$]
  \item $\kg$ is quantum polynomial-time,
  \item $\Pr \bigl[ (sk,pk)\in R: (sk,pk)\gets \kg(\usecpar) \bigr] \geq 1 -
    \negl(\secpar)$.
  \item for all quantum polynomial-time algorithm $\attack$, there is a quantum
    polynomial-time algorithm $\dkg$ such that
    \begin{equation*}
      \Bigl| \Pr \bigl[ \attack(pk) = 1: (sk,pk)\gets \kg(\usecpar) \bigr] -
      \Pr \bigl[ \attack(\dualkey)=1: \dualkey \gets \dkg(\usecpar) \bigr]
      \Bigr| \leq \negl(\secpar),
    \end{equation*}
    and
    \begin{equation*}
      \Pr \Bigl [ \dualkey \in L_R: \dualkey \gets \dkg(\usecpar)
      \Bigr] \leq \negl(\secpar).
    \end{equation*}
  \end{itemize}
  \label{def:dual-gen}
\end{definition}

\begin{theorem} $\kg$ in $\gaids$ is a dual-mode key generator, if
  Assumption~\ref{asn:pra} holds.
  \label{thm:ga-dual-gen}
\end{theorem}

\begin{proof} We construct $\dkg$ as follows:

  \begin{enumerate}
  \item choose $(G,S,\alpha)$ to be $\pg(\usecpar)$;
  \item sample $s,t\gets S$ uniformly;
  \item output $\dualkey:= (s,t)$.
  \end{enumerate}

  By Assumption~\ref{asn:pra}, it follows that
  \begin{equation*}
    \Bigl| \Pr \bigl[ \attack(pk) = 1: (sk,pk)\gets \kg(\usecpar) \bigr] -
    \Pr \bigl[ \attack(\dualkey)=1: \dualkey\gets \dkg(\usecpar) \bigr] \Bigr|
    \leq \negl(\secpar)\, .
  \end{equation*}

  In addition,
  \begin{equation*}
    \Pr \bigl[\, pk\in L_R: pk \gets \dkg(\usecpar) \,\bigr] = \frac{|G|}{|S|}
    \leq \negl(\secpar). \qedhere
  \end{equation*}

\end{proof}

\begin{figure}[htb!]
  \centering
  \begin{protocol}
    \ul{Fiat-Shamir transformation}\\[1em]
    Let $\secpar$ be the security parameter, $\lencom,\lench,\lenr$ be the
    length of the commitment, challenge and response respectively.
    Choose a hash function $H: \bit^* \to \bit^{\lenr}$.
    Given an identification scheme $\ids$, we construct a signature scheme
    $\fssign$ as follows:
    \begin{itemize}[label=$\bullet$]
    \item $\skg$: run $(sk, pk) \gets \kg(\usecpar)$. Output $(sk,pk)$.
    \item $\sign$: for input message $m$ and $sk$,
      \begin{enumerate}[label=\roman*)]
      \item compute $I \gets \prover(sk,pk)$, $c = H(pk,m,I)$, and
        $r\gets \prover(sk,pk,I,c)$.
      \item output $\sigma: = I \| r$.
      \end{enumerate}
    \item $\vrfy$: for input $pk$ and $(m,\sigma)$,
      \begin{enumerate}[label=\roman*)]
      \item parse $\sigma$ as $I \| r$.
      \item compute $c := H(pk,m,I)$.
      \item accept if $\verifier(I,c,r) = 1$.
      \end{enumerate}
    \end{itemize}
  \end{protocol}
  \caption{Fiat-Shamir transformation}
  \label{fig:fiat-shamir}
\end{figure}

\begin{theorem}(Adapting~\cite[Corollary 33]{UnruhAC17}) If an identification
  scheme $\ids$ has correctness, HVZK, statistical soundness, unpredictable
  commitments, $\lench$ is superlogarithmic, and $\kg$ is a dual-mode key
  generator.
  Then in \QRO, $\fssign$ obtained from Fiat-Shamir transform (Construction in
  Fig.~\ref{fig:fiat-shamir}) is weakly unforgeable.
  If $\ids$ has unique responses, the signature scheme is strongly unforgeable.
  \label{thm:dual-fs}
\end{theorem}

Since $\gaids$ has all these properties, we can instantiate Construction
$\fssign$ with $\gaids$.
Call the resulting signature scheme $\gafssign$.

\begin{corollary} $\gafssign$ is strongly unforgeable, if
  Assumption~\ref{asn:pra} holds.
  \label{cor:ga-fs-sign}
\end{corollary}

Note that $\gafssign$ is much more efficient than $\gasign$. In
particular, $\gafssign$ only invokes the underlying
$(\prover,\verifier)$ once as opposed to superpolylogarithmic times in
$\gasign$.

\subsection{Quantum-secure pseudorandom functions based on \PRA}
\label{subsec:pra_prf}

Finally, we discuss how to construct quantum-secure
pseudorandom functions using the $\PRA$
assumption.  Basically, we will show that we can instantiate the GGM
construction~\cite{Goldreich:1986aa} using the $\PRA$ assumption.  To do this,
we need to first discuss constructing pseudorandom generators.

\paragraph{(Keyed) pseudorandom generators.}
We already have mentioned that we can construct a PRG $\Gamma:S\times G\to S\times
S$, given by
\[
\Gamma(s, g):=(s, g\cdot s).
\]

In fact, we may modify this construction slightly to obtain a form of PRG with
much better stretching almost for free as follows.  For $s\in S$, we define
$\Gamma_s:G\to S$ by
\[
\Gamma_s(g):=g\cdot s.
\]

This can be considered as a `keyed PRG', where $s$ is a public key for the PRG
instance $\Gamma_s$, and this instance stretches the seed $g\gets G$ to $g\cdot
s$.  Such notion of a keyed PRG is informally given in~\cite{Huang:2012aa}, but
surely this notion was used implicitly in many works previously.  We may give a
formal definition of this notion as follows.
\begin{definition}[Keyed PRG]
A \emph{keyed pseudorandom generator}, or a keyed PRG, is a pair of probabilistic
polynomial-time algorithms $(\kg, \prg)$:
\begin{itemize}
\item Key generator: $k\gets\kg(1^\lambda)$ generates a \emph{public} key
$k\in\cK$ describing an instance of the keyed PRG.
\item Pseudorandom generator: given $k$ sampled by $\kg(1^\lambda)$,
$\prg_k:\cX\to\cY$ stretches a uniform element $x\gets\cX$ to produce an element
$\prg_k(x)\in\cY$.  Note that this $\prg$ algorithm is required to be
deterministic.
\end{itemize}
In the above, $\cK$ is the key space of the keyed PRG, and $\cX$, $\cY$ are the
domain and the codomain of the keyed PRG, respectively.  They are implicitly
parametrized by the main parameter $\lambda$.  Also, it is required that
$|\cY|>|\cX|$.
\end{definition}

\begin{definition}[Security of a keyed PRG]
  We say that a keyed PRG, $\Gamma=(\kg, \prg)$, is \emph{secure}, if for any
  quantum polynomial-time adversary $\attack$, we have
  \begin{equation}
    \begin{split}
      \adv_\Gamma^{\mode{prg}}(\attack) & := \Bigl| \Pr \bigl[
      \attack(k, \prg_k(x))=1 \,:\, x\gets\cX, k\gets\kg(1^\lambda) \bigr] \\
      &\qquad -\Pr \bigl[ \attack(k, y)=1 \,:\, y\gets\cY, k\gets\kg(1^\lambda)
      \bigr] \Bigr| \leq\negl(\lambda).
    \end{split}
  \end{equation}
\end{definition}

Again, it is immediate that $\PRA$ assumption implies that $g\mapsto g\cdot s$
is a secure keyed PRG, where $s$ is the key and $g$ is the seed.

\paragraph{Doubling keyed PRGs.}

The keyed PRG $\Gamma_s$ that we have described above is of form $\Gamma_s:G\to
S$.  While $|S|\gg|G|$, having $S$ which might `look different' from $G$ can be
inconvenient for some applications, for example, constructing a PRF via the GGM
construction.  So, here we would like to construct a `doubling' keyed PRG out of
the previous construction, using randomness extraction.

The idea is simple: $\Gamma_s(g)$ would look uniform random over $S$ for average
$s$, so we can use a randomness extractor to produce a pseudorandom bit string of
enough length, and use that to sample two group elements of $G$.  Overall, the
construction would be of form $G\to G\times G$, while the PRF key would include
not only the point $s\in S$ but also the random seed for the randomness
extraction.  For concreteness, we may use the Leftover Hash Lemma
(LHL)~\cite{HILL99}, but in fact any strong randomness extractor would be all
right.

More concretely, let $R_G:\{0,1\}^p\to G$ be the sampling algorithm for the group
$G$ which samples a random element of $G$, (statistically close to) uniform.  Note
that this $R_G$ is required for our group $G$.  In fact, Babai~\cite{Babai:1991aa}
gives an efficient Monte Carlo algorithm for sampling a group element of a finite
group in a very general setting which is applicable to all of our instantiations.

Let $\mathcal{H}=\{h:S\to\{0,1\}^r\}$ be a family of 2-universal hash functions,
where $r$ is sufficiently smaller than $\log |S|$.  LHL implies, informally, that
$(h, h(s))$ and $(h, u)$ are statistically indistinguishable, when
$h\gets\mathcal{H}$, $s\gets S$, $u\gets \{0,1\}^r$ are uniform and independent.
Let us assume $\log |S|$ is large enough so that we can take $r=2p$.

Then, we may construct a doubling keyed PRG $(\kg, \prg)$ as follows:
\begin{itemize}
\item  Choose public parameters $\params(G, S, \alpha)$ to be $\pg(\usecpar)$.
\item Key generator: $\kg(1^\lambda)$ samples $s\gets S$, $h\gets\mathcal{H}$,
  and outputs $k:=(s, h)$.
\item Pseudorandom generator: $\prg_k(g):=(R_G(r_0), R_G(r_1))$, where $r_0$ and
  $r_1$ are the left half and the right half of $h(g\cdot s)$, respectively.
\end{itemize}
In short, this keyed PRG stretches the seed $g\gets G$ to $g\cdot s$, and extracts
a pseudorandom bit string of length $2p$, and use that to sample two
independent-looking group elements.  The security of this construction comes from
the $\PRA$ assumption and the Leftover Hash Lemma.

\paragraph{Pseudorandom functions.}

Of course, the notion of a PRF~\cite{Goldreich:1986aa} is well-known and
well-established.  Here, following Maurer and Tessaro~\cite{MT08}, we are going to
extend the notion of PRF somewhat so that it may also have an extra `public key'
part.

\begin{definition}[Pseudorandom function]
  A pseudorandom function (PRF) is a polynomial-time computable function $f$ of
  form $f:\cP\times\cK\times\cX\to\cY$.  We call the sets $\cP$, $\cK$, $\cX$,
  $\cY$ as the public-key space, the key space, the domain, and the codomain of
  $f$, respectively.

  We would often write $f(p, k, x)$ as $f_p(k, x)$.
\end{definition}
Note that we may regard an `ordinary' PRF as a special case of above where it has
a trivial, empty public key.

In this paper, we consider quantum-secure PRFs~\cite{Zhandry:2012ac}, or,
sometimes called QPRFs, whose security is defined as follows.

\begin{definition}[Security of a PRF]
  Let $f:\cP\times\cK\times\cX\to\cY$ be a PRF.
  We say that $f$ is \emph{quantum-secure}, if for any quantum polynomial-time
  adversary $\attack$ which can make quantum superposition queries to its oracle,
  we have the following:
  \begin{equation*}
    \adv^{\mode{prf}}_f(\attack) := \Bigl| \Pr \bigl[
    \attack^{f_p(k, \cdot)}(p)=1 \bigr]
    - \Pr \bigl[ \attack^{\rho}(p)=1 \bigr] \Bigr| = \negl(\secpar),
  \end{equation*}
  where $p\gets\cP$, $k\gets\cK$, $\rho\gets\cY^\cX$ are uniformly and
  independently random and $\secpar$ is the security parameter.
\end{definition}

Suppose we have a secure, doubling keyed PRG $\Gamma=(\kg, \prg)$ where $\prg_s$
is of form
\[
\prg_s:\cK\to\cK\times\cK.
\]
Writing the first component and the second component of $\prg_s(k)$ as $f_s(k, 0)$
and $f_s(k, 1)$, we obtain a PRF $f$ of form
\[
f:\cS\times\cK\times\{0,1\}\to \cK.
\]
Here, $\cS$ is the public-key space of $f$, which is the key space of the keyed
PRG $\Gamma$.  The key space of $f$ is $\cK$, and the domain and the codomain of
$f$ are $\{0,1\}$ and $\cK$, respectively.

Moreover, we can immediately see that the security of the one-bit PRF $f$ is
exactly equivalent to the security of $\Gamma$ as a keyed PRG.  In fact, we can
say that the security of $f$ is just a re-statement of the security of $\Gamma$.

Now we may apply the GGM construction to $f$ to define the following PRF
$\ggm[f]:\cS\times\cK\times\{0,1\}^l\to\cK$, where
\[
\ggm[f]_s(k, x_1\dots x_l):= f_s(\dots f_s(f_s(k, x_1), x_2), \dots, x_l).
\]

In fact, the above is the same as the cascade construction for the one-bit PRF $f$.

When we instantiate the GGM construction using an ordinary PRG, or when we
instantiate the cascade construction using an ordinary PRF (without the public-key
part), the quantum security is already
established~\cite{Zhandry:2012ac,Song:2017aa}.  The only difference is that here
we instantiate the construction using a keyed PRG, or, equivalently, a one-bit PRF
with a public key.

Following~\cite{Zhandry:2012ac} or~\cite{Song:2017aa}, we can define a version of
oracle security for such a PRF with a public key.

\begin{definition}[Oracle security of a PRF]
  Let $f:\cP\times\cK\times\cX\to\cY$ be a PRF.  We say that $f$ is
  \emph{oracle-secure with respect to an index set $\cI$}, if for any quantum
  polynomial-time adversary $\attack$ which can make quantum superposition queries to its
  oracle, we have the following:
  \begin{equation*}
    \adv^{\mode{os-prf}}_{f, \cI}(\attack) := \Bigl| \Pr \bigl[
    \attack^{O_0}(p)=1 \bigr] - \Pr \bigl[ \attack^{O_1}(p)=1
    \bigr]\Bigr|=\negl(\secpar),
  \end{equation*}
  where the oracles $O_0, O_1$ are defined as
  \begin{align*}
    O_0(i, x):=f_p(\kappa(i), x),\quad  O_1(i, x)&:=\rho(i, x),
  \end{align*}
  and $p\gets\cP$, $\kappa\gets\cK^\cI$, $\rho\gets\cY^{\cI\times\cX}$ are chosen
  uniform randomly and independently.
\end{definition}

And, as in~\cite{Song:2017aa}, we show that if a PRF with a public key is secure,
then it is also oracle-secure.

\begin{theorem}
Let $f:\cP\times\cK\times\cX\to\cY$ be a PRF.  Suppose that it is secure as a
PRF.  Then, it is also oracle-secure.
\end{theorem}
\begin{proof}
Here is a brief sketch of the proof.  We are going to use the notion of relative
(oracle) indistinguishability introduced in~\cite{Song:2017aa}.  Our random oracle
$H$ would be a very simple one, $H:\{\ast\}\to\cP$, where $\{\ast\}$ is the
singleton set containing only one element.  Given this $p=H(\ast)\gets \cP$, we
define two distributions $D_0, D_1$ of functions of form $\cX\to\cY$.
\begin{itemize}
	\item $D_0$: to sample a function $g$ from $D_0$, sample $k\gets\cK$, and
	define
	\[
	g(x):=f_p(k, x).
	\]
	\item $D_1$: to sample a function $g$ from $D_1$, simply sample a uniform
	random function $g:\cX\to\cY$.
\end{itemize}

Then, the security of $f$ is in fact equivalent to indistinguishability of $D_0$
and $D_1$ relative to the simple oracle $H$.  Again according
to~\cite{Song:2017aa}, when two function distributions are indistinguishable
relative to $H$, then they are oracle-indistinguishable relative to $H$.  We can
also observe that this is equivalent to the oracle security of $f$ defined as
above.
\end{proof}

Finally, the security of the GGM construction comes from the oracle security.

\begin{theorem}
Suppose that $f:\cS\times\cK\times\{0,1\}\to \cK$ is a secure PRF.  Then, the GGM
construction $\ggm[f]:\cS\times\cK\times\{0,1\}^l\to\cK$ is also secure.
\end{theorem}
\begin{proof}
The proof is essentially identical to that of Zhandry~\cite{Zhandry:2012ac} or
Song and Yun~\cite{Song:2017aa}; since $f$ is secure as a PRF, it is also
oracle-secure.  This allows the same hybrid argument in the security proof for GGM
in~\cite{Zhandry:2012ac}, or the security proof for the cascade construction
in~\cite{Song:2017aa}.
\end{proof}

\subsection{Bit commitment schemes based on \OWA and \PRA}
\label{sec:supp}

Based on \OWA (Assumption~\ref{asn:owa}), Brassard and
Yung~\cite{BY90}
describe a bit commitment scheme, which we can easily adapt and
instantiate it with non-abelian group actions. 

\begin{figure}[htb!]
  \centering
\begin{protocol}
  \ul{(Perfectly hiding and computationally binding) Bit commitment}\\[1em]
  Suppose Alice wants to commit to Bob a bit $b\in \{0, 1\}$.
  \begin{enumerate}
  \item Bob samples $s_0\in S$ and $g\in G$, and computes $s_1=g\cdot s_0$.
  \item Bob convinces Alice that $s_0$ and $s_1$ are in the same orbit, using
    the identification protocol in Section~\ref{sec:owa_id-def}.
  \item To commit to $b\in \{0, 1\}$, Alice samples $h\in G$, computes
    $t=h\cdot s_b$, and sends $t$.
  \item To open $b\in \{0, 1\}$, Alice sends $h$ to Bob, and Bob
    verifies that $h\cdot s_b=t$.
  \end{enumerate}
\end{protocol}
  \caption{A bit commitment scheme based on \OWA}
  \label{fig:bit-commit}
\end{figure}

We briefly argue how the security conditions are met. A formal proof
(against both classical and quantum attacks) can be obtained along the
same line. Let $b'=1-b$.
\begin{itemize}
\item \emph{Binding}\footnote{Here we do not consider more
    sophisticated binding notions in the quantum setting
    (Cf.~\cite{UnruhEC16}).}: in order to change her mind, Alice needs
  to compute $h'$ such that $h' \cdot s_{b'}=t$. If she can do that, given
  that she already knows $h \cdot s_b=t$, she can compute $g$. This violates
  the one-way assumption.
\item \emph{Hiding}: from Bob's viewpoint, since $s_0$ and $s_1$ are
  in the same orbit, whichever bit Alice commits, the distributions of
  $t$ are the same.
\end{itemize}

Let us then examine an alternative route to bit commitment based on the 
\PRA assumption. Consider the function below:
\begin{align*}
  T: S \times G & \to S\times S \\
     (s,g) & \mapsto (s, g \cdot s) \, .
\end{align*}
It is easy to see that this gives a quantum-secure \emph{pseudorandom
  generator} (PRG) based upon Assumption~\ref{asn:pra}, assuming that
the size $\abs{S}$ is larger than $\abs{G}$.  Therefore, after we
apply the Blum-Micali amplification to increase the expansion to
triple, we can plug it into Naor's commitment~\cite{Naor91} and get a
\emph{quantum computationally hiding} and \emph{statistically binding}
commitment.

\begin{theorem} There is a \emph{perfectly hiding} and
  \emph{computationally binding} bit commitment, if
  Assumption~\ref{asn:owa} holds; there is a \emph{quantum
    computationally hiding} and \emph{statistically binding}
  commitment scheme, if Assumption~\ref{asn:pra} holds.
  \label{thm:comm}
\end{theorem}

\paragraph{Acknowledgement.} Y.Q. would like to thank Joshua
A. Grochow for explaining the results in \cite{FGS19} to him. Y.Q. was
partially supported by Australian Research Council
DE150100720. F.S. was partially supported by the U.S. National Science
Foundation under CCF-1816869. 
Any opinions, findings, and conclusions
or recommendations expressed in this material are those of the
author(s) and do not necessarily reflect the views of the National
Science Foundation.
A.Y. was supported by Institute of Information \& Communications Technology Planning
\& Evaluation (IITP) grant funded by the Korea government (MSIT) (No.\
2016-6-00598, The mathematical structure of functional encryption and its
analysis).

\bibliography{group-action}
\bibliographystyle{alpha}

\end{document}